\documentclass[%
 reprint,
superscriptaddress,
 amsmath,amssymb,
 aps,
prb,
]{revtex4-1}
\usepackage{blindtext,float}
\usepackage{graphicx,xcolor}
\usepackage{dcolumn}
\usepackage{bm}
\usepackage{mathtools}
\usepackage{gensymb}
\usepackage{amsmath}
\usepackage{qcircuit}
\usepackage{dsfont,soul,xcolor}
\usepackage[caption=false]{subfig}
\DeclarePairedDelimiter\bra{\langle}{\rvert}
\DeclarePairedDelimiter\ket{\lvert}{\rangle}

\usepackage[export]{adjustbox}
\makeatletter
\newcommand*{\addFileDependency}[1]{
  \typeout{(#1)}
  \@addtofilelist{#1}
  \IfFileExists{#1}{}{\typeout{No file #1.}}
}
\makeatother

\DeclareMathSymbol{\shortminus}{\mathbin}{AMSa}{"39}

\usepackage{amsthm}
\theoremstyle{plain}

\newtheorem{theorem}{Theorem}[section]
\begin{document}


\title{Creation and annihilation operators for 2D non-abelian anyons}
\author{Nicetu Tibau Vidal}
\affiliation{%
Clarendon Laboratory, Department of Physics, University of Oxford, Oxford OX1 3PU, United Kingdom
}%
\affiliation{%
QICI Quantum Information and Computation Initiative, Department of Computer Science,
The University of Hong Kong, Pok Fu Lam Road, Hong Kong
}%
\author{Lucia Vilchez-Estevez}
\affiliation{%
Clarendon Laboratory, Department of Physics, University of Oxford, Oxford OX1 3PU, United Kingdom
}%




\begin{abstract}
We define creation and annihilation operators for any 2D non-abelian anyon theory by studying the algebraic structure from the anyon diagrammatic formalism. We construct the creation operators for Fibonacci anyons explicitly. We obtain that a single creation operator per particle type is not enough; we need an extra creation operator for every alternative fusion channel. We express any physically allowed observable in terms of these creation and annihilation operators. Finally, we express the 2D Fibonacci Hubbard Hamiltonian in terms of the Fibonacci creation and annihilation operators, and we comment on developing methods for simulation based on these creation and annihilation operators.  

\end{abstract}

\maketitle


\section{Introduction} \label{sec:intro}
Anyons are postulated quasiparticle excitations in two-dimensional systems \cite{Wilczek1982}. They have a topological nature and exotic exchange statistics \cite{Leinaas1977,Wilczek1982,Nayak2008,Haldane,ChernSimons,Goldin1995}, which differentiate them from bosons and fermions. We call them topological particles or phases of matter because the geometry of space-time or the distance between them does not change the result of the relevant operations. These topological properties make anyons systems a promising platform for quantum information processing \cite{Kitaev2003computation,KitaevFreedman,Freedman2001,eisert,bonderson}. Topological quantum computing tries to exploit these features to have a robust computation against error due to local perturbations and noise by the environment. However, the experimental discovery of such systems has remained elusive so far \cite{fibonacci,ising3,ising1,Beenakker2020,Josephson,fqhe}. 

Information processing with topological systems has been one of the main attractions to the study of anyonic theories. We build on the recent information-theoretic perspective on anyons \cite{Rowell,BondersonPhd,Beer,KitaevFormalism,pachos,Xu2022,Shapourian2020}. Nevertheless, anyons can also be very intriguing from a more foundational point of view. The notion of subsystems and locality in quantum information theory is crucial to understanding interactions between different systems. In a qubit theory, for example, we use the tensor product structure to describe systems consisting of multiple subsystems. Two non-abelian anyons can merge (fuse) together to different anyonic charges depending on the fusion channel. Therefore, to completely describe an anyonic quantum system, we need to know all the charges that make up the system and how they fuse with each other. This means there is no such thing as a tensor product between two subsystems since we need that extra bit of information on the overall charge of the composed system. 

There is a gap in the literature when talking about a creation and annihilation operator algebra for non-abelian anyons in 2D. Bosons and fermions have well-defined annihilation operators, so it is natural to look for them in anyon theories too. For anyons in one spatial dimension, the creation and annihilation operators have been found \cite{1dbethe}. We believe that in the 2D case, the difficulty of defining modes (or subsystems) and the topological charge superselection rule are the main reasons for this literature gap. The latter is an interesting characteristic of anyon theories that ensures operators will only be physical observables when the total topological charge is conserved. 

In this work, we define an anyonic mode as simply connected sub-regions with boundaries of our two-dimensional space. We can then map the subsystem structure to the level of simply connected regions with the help of the planar representation of anyons \cite{BondersonPhd}. Using the diagrammatic approach for anyons, we can find the candidates for annihilation operators within the operators left invariant by local transformations on the rest of the system.

\textcolor{blue}{The rest of the paper is organised as follows: In Section \ref{sec:adiag} we review the diagrammatic formalism of non-abelian anyons. In Section \ref{sec:operators}, we define the notion of subsystem in anyonic systems and present our result of the anyonic creation and annihilation operators. In Section \ref{sec:examples}, we obtain the creation operators for Fibonacci anyons. In Section \ref{sec:hamiltonian}, we express the Fibonacci Hubbard Hamiltonian in terms of the creation and annihilation operators. In Section \ref{sec:discussion}, we discuss our results.}

\section{Anyon diagrams}
\label{sec:adiag}

An anyon is a quasiparticle that can exist in two-dimensional systems. We can think of putting two anyons together to create a new particle. This process is known as \emph{fusion}. Two particles $a$ and $b$ can be fused to $c$, which will read as $a\times b = b\times a = c$. However, in non-abelian anyon theories, it is possible to find different outcomes for the same fusion channel; in this case, we write:
\begin{equation}
    a\times b =b \times a = \sum_c N_{ab}^c c,
    \label{eq:fusion}
\end{equation}
where $N_{ab}^c$ are the fusion multiplicities. They indicate the number of different ways in which $a$ and $b$ can fuse to $c$. There is a trivial anyon $e$, the vacuum or the identity. This particle satisfies the property $N_{e a}^b=\delta_{ab}$. Every particle $a$ also has its own antiparticle $\bar{a}$ such that $N_{ab}^e=\delta_{b \bar{a}}$.

We can write an orthonormal complete set of states for $n$ anyons as a fusion tree as in Figure \ref{fig:states}. If any of the $N^{a_{i+1}}_{a_{i-1}a_i}=0$, then the fusion is not allowed, and the diagram is zero. The corresponding bras $\bra{\psi_i}$ are obtained by doing the Hermitian conjugate, which is equivalent to flipping the diagram along a horizontal axis.

\begin{figure}[H]
    \centering
    \includegraphics[width=0.45\textwidth]{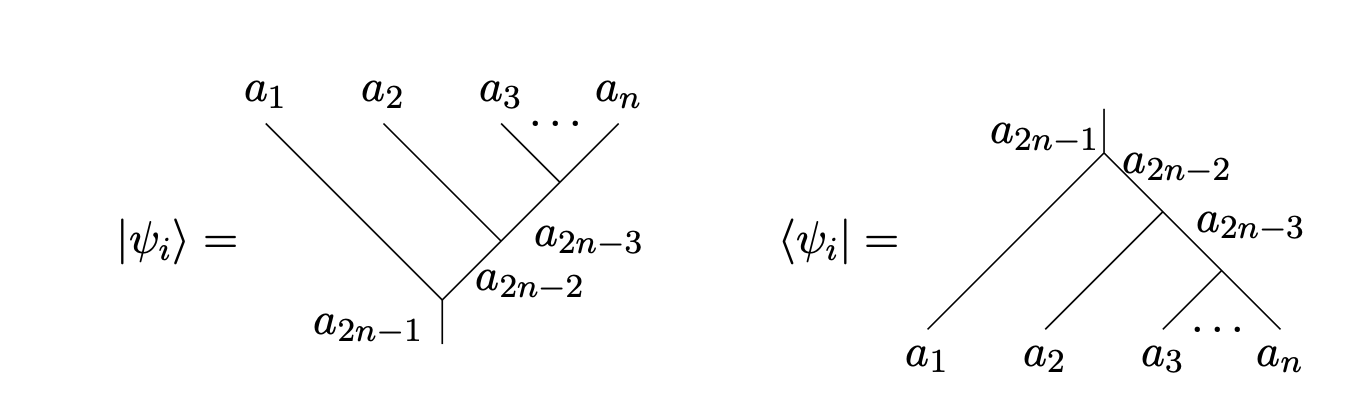}
    \caption{Basis $\ket{\psi_i}$ and its conjugate $\bra{\psi_i}$ of an $n$-anyon system. All vertices are allowed fusion channels.}
    \label{fig:states}
\end{figure}

When utilizing the diagrammatic algebra, we will always set the time direction vertically and upwards and assume that all particles move forward in time. We can interpret a particle going back in time as its antiparticle moving forward in time.

One can use the basis states to build arbitrary operators in the same way we do when using kets and bras. A diagram with lines pointing both upwards and downwards can be interpreted as operators that take as input the particle lines coming in from the bottom and give as output the lines coming out the top. The lines coming in from the bottom are the bra part of the operator, and the lines pointing out are the ket part. For instance, we can write a general operator as in Figure \ref{fig:operator}.

\begin{figure}[H]
    \centering
    \includegraphics[width=0.4\textwidth]{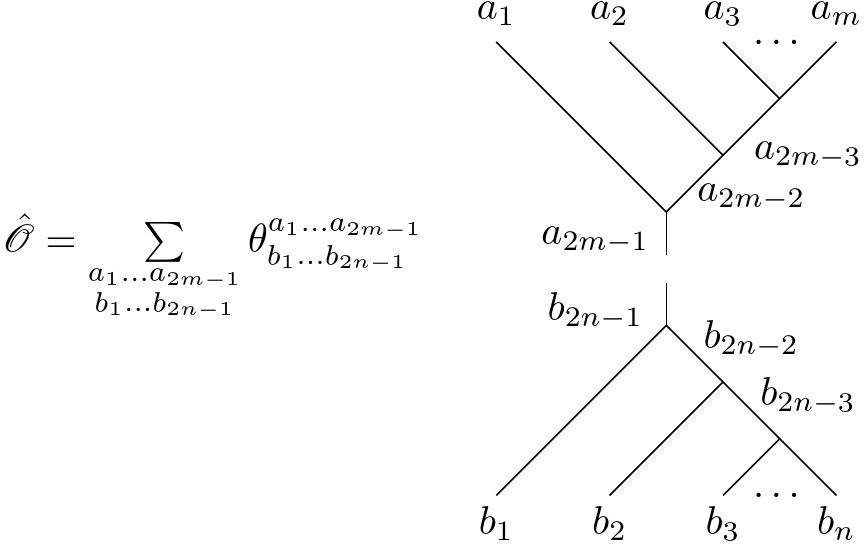}
    \hfil
    \caption{General operator with $n$ inputs and $m$ outputs. The coefficients are arbitrary.}
    \label{fig:operator}
\end{figure}

An important thing to keep in mind is that operators will only be physical observables when the total charge is conserved. In Figure \ref{fig:operator}, this would mean that $a_{2m-1}=b_{2n-1}$ and thus the diagram would be connected. This is a direct consequence of the strong superselection rule that exists in anyonic systems \cite{rehrenSSR}. It is not possible to implement an operator that changes the overall topological charge of the system. 

We will be particularly interested in one family of non-abelian anyons: Fibonacci anyons\cite{fibonacci}. The Fibonacci model is perhaps the simplest non-abelian example and has only two particle types, the vacuum or trivial anyon $e$ and the Fibonacci anyon $\tau$. The only non-trivial fusion rule of this theory reads
\begin{equation}
    \tau \times \tau = e+\tau.
    \label{eq:fibonacci}
\end{equation}

One can convert between bases associated with different fusion trees by using the $F$-matrices shown in Figure \ref{fig:fusion}. In the Fibonacci theory, the only nontrivial $F$-matrix is $[F_{\tau}^{\tau \tau \tau}]=\begin{pmatrix}
\phi^{-1} & \phi^{-1/2}\\
\phi^{-1/2} & -\phi^{-1}
\end{pmatrix}$.

\begin{figure}[H]
    \centering
    \includegraphics[width=0.3\textwidth]{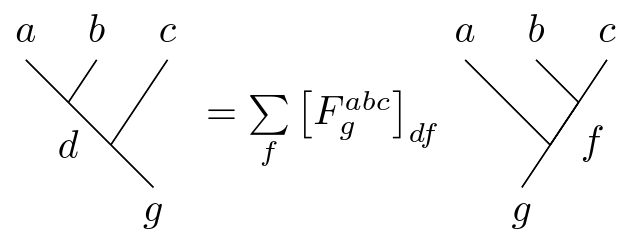}
    \caption{The $F$-matrix defines a change of basis.}
    \label{fig:fusion}
\end{figure}

Further, remember that exchanging two anyons in a multiplicity-free theory results in a phase factor that depends on their overall charge. For Fibonacci anyons, there are two non-trivial exchanging diagrams: (i) when two $\tau$ anyons fuse to the identity and (ii) when two $\tau$ anyons fusing to the $\tau$ anyon. The phases are $R_e^{\tau \tau}=e^{-4\pi i/5}$ and $R_{\tau}^{\tau \tau}=e^{3\pi i/5}$, respectively.

\section{Anyonic annihilation operators}\label{sec:operators}
To define anyonic annihilation operators, we first need a notion of \emph{modes} that can be excited \cite{modes1,bosonmode,Friismodeent}. Usually, these modes refer either to momentum in quantum field theory or to lattice sites in the usual Ising chain models. For simplicity, we prefer to keep the number of modes finite and use the notion of mode as a lattice site. 

We want to identify a simply connected sub-region with boundaries of our 2D space as a single mode where the different anyon types can be excited. We consider that the complete system consists of a finite number $N$ of such regions glued along their boundaries; see Figure \ref{fig:locality}. Therefore, we use a finite 2D lattice populated by the different anyon particle types of the theory.

\begin{figure}[H]
    \centering
    \includegraphics[width=0.25\textwidth]{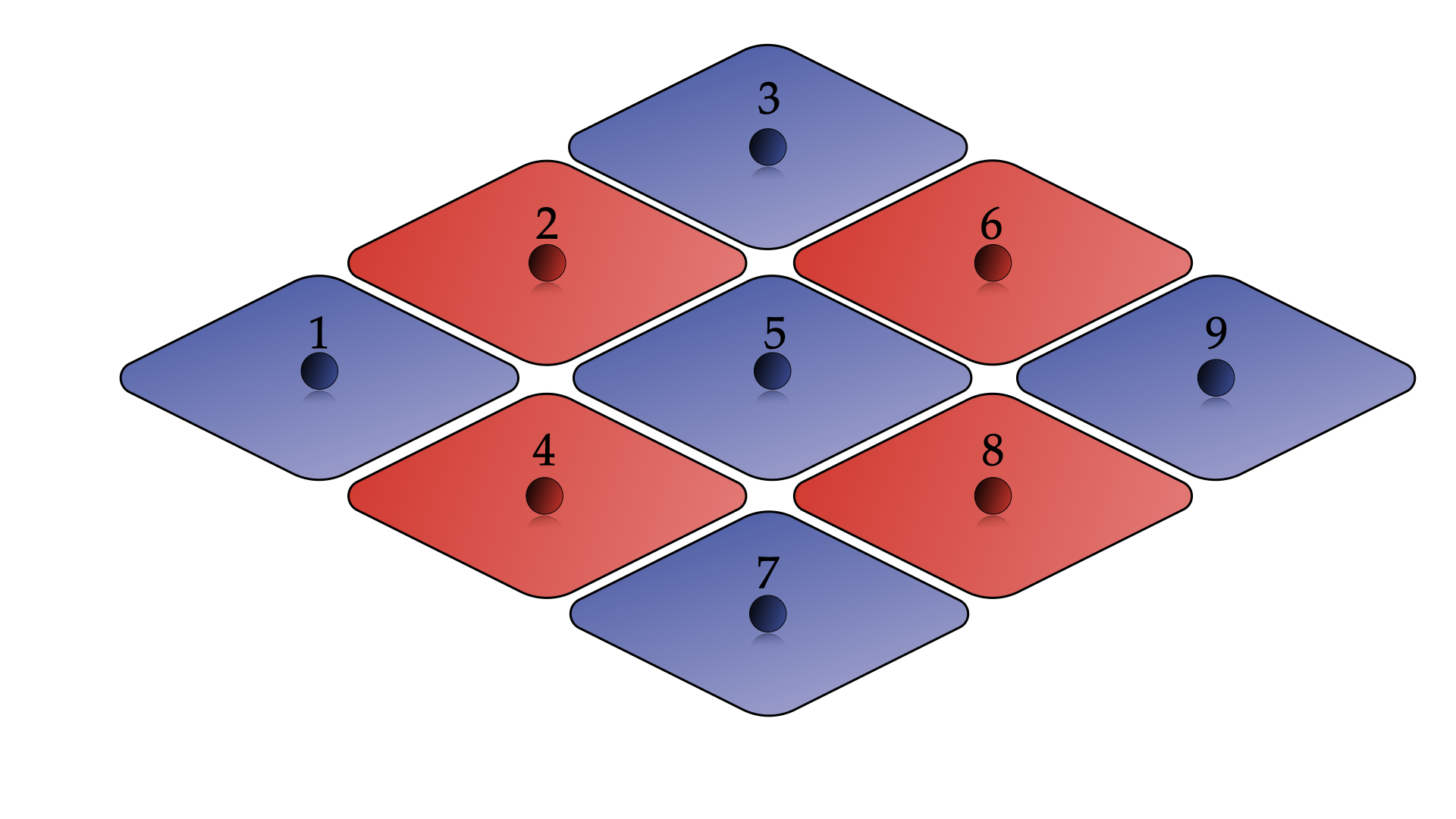}
    \caption{Partition of the plane in different subregions that are associated with modes. Notice that the union of regions $1,2,3,4,6,7,8$ and $9$ is not a simply connected region. We would not consider it a valid subsystem.}
    \label{fig:locality}
\end{figure}

To consider the annihilation operators, we want to identify the modes as the elementary subsystems in the theory. We want to understand how to map the subsystem structure at the level of simply connected regions in the 2D manifold to planar diagrams. Notice that there are different ways to glue the boundaries between the regions to compose them into larger simply connected regions.  By defining a fusion order, these different planar representations correspond to different partitions of the systems given by the planar canonical basis of the anyon theory.   
\begin{figure}[H]
    \centering \includegraphics[width=0.47\textwidth]{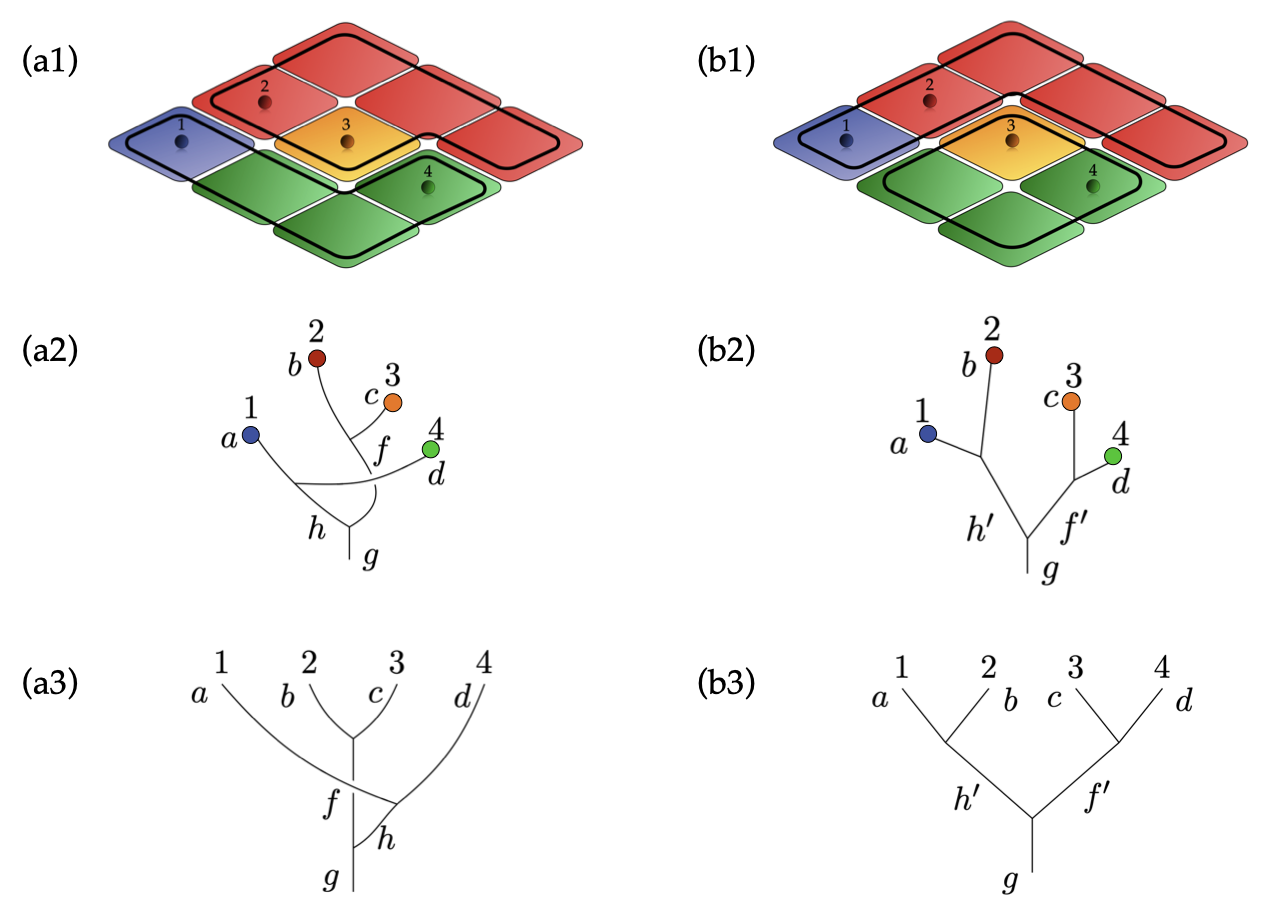}
    \caption{Different planar representations for different compositions of regions. Sub-figure (a1) indicates that we are first fusing anyon $1$ (blue) with anyon $4$ (green) and anyon $2$ (red) with anyon $3$ (orange). In (a2), we express such a system in the diagrammatic form, which is equivalent to the planar representation in sub-figure (a3). In the right column, we have the same, but when we fuse anyon $1$ with anyon $2$ and anyon $3$ with anyon $4$.}
    \label{fig:planarreps}
\end{figure}

As we said, to define annihilation operators, it is helpful to understand each mode as an elemental subsystem. We do this step for two reasons. First, we would like to use the conceptualization of mode subsystems done in the literature of fermionic and bosonic annihilation operators \cite{bosonmode,Friismodeent}. Secondly, if we consider the mode as a subsystem, we can find candidates for annihilation operators within the operators left invariant by transformations local to the system of the rest of the modes. 

Even though there is not a clear notion of a general local operator in anyonic systems, there is the notion of physical local unitaries and observables, as shown in Figure \ref{fig:localunitaries}. The conservation of anyonic charge allows to have well-behaved physical local operators such as unitaries and observables. However, for creation and annihilation operators, we expect the conservation of charge to be violated from what we observe in fermionic systems\cite{Nthesis}. 

If a system consists of modes $M=\{1,\dots,m+1\}$, we can say that a candidate local operator in mode $i\in M$ is an operator $\hat{O}$ such that is invariant under the action of all physical local unitaries in the modes $M\backslash \{i\}$. In equation form that reads as: $\hat{O}$ is a candidate local operator on mode $i\in M$ if and only if
\begin{eqnarray} 
\hat{U}^\dagger_{M\backslash \{i\}} \cdot \hat{O} \cdot \hat{U}_{M\backslash \{i\}}=\hat{O} 
\label{eq:local}
\end{eqnarray}

for all $\hat{U}_{M\backslash\{i\}}$ being an allowed local unitary in modes $M\backslash\{i\}$. This is a natural property that a local operator must satisfy. If one works within the Heisenberg picture of quantum mechanics, it is clear that, indeed, when evolving a local operator in $A$ with a physically allowed local unitary in $B$, then the local operator in $A$ must be left invariant. Moreover, it is not difficult to check that the conditions in Figure \ref{eq:local} give that the collection of all candidate local operators in $i$ form an algebra under the usual sum and operator multiplication, and $\mathbb{C}$ as scalars. So, we can say that we have an abstract definition of the algebra of candidate local operators in mode $i$.  

\begin{figure}[H]
    \centering
    \includegraphics[width=0.4\textwidth]{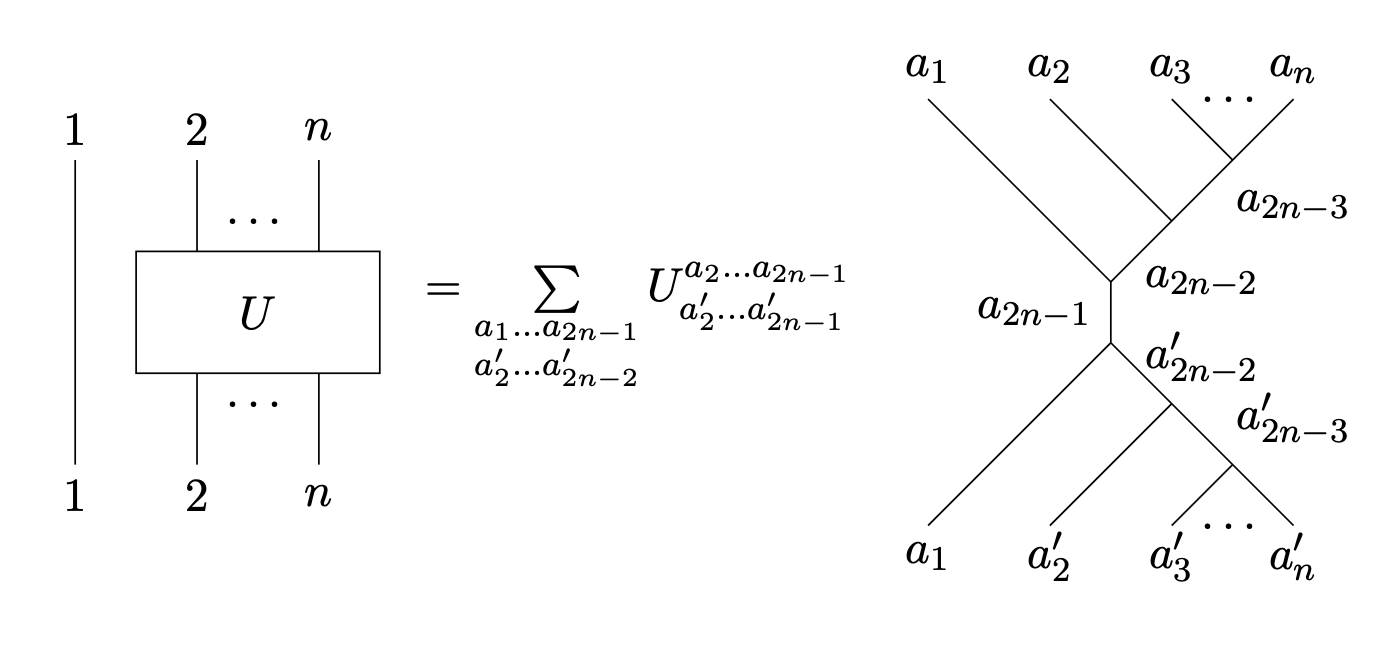}
    \caption{General form of a unitary local in $M\backslash \{1\}$}
    \label{fig:localunitaries}
\end{figure}

Using the diagrammatic approach for anyons, we can characterise the allowed local unitaries and explore the candidate local operators for any given mode. In Figure \ref{fig:localunitaries}, we show how an allowed local unitary looks in diagrammatic form. We solve equation \ref{eq:local} that defines candidate local operators using the diagrammatic formalism, and we find the general form of a candidate local operator on an anyonic mode. For simplicity, we show it here for the first mode. We express the general form of a candidate local operator on mode $1$ in terms of linear combinations of the elements of a canonical basis: 

\begin{eqnarray}
O_1=\sum_{\substack{a,a',b_0 \\d=a\times b_0 , d'=a'\times b_0}} c_{a,a',b_0,d,d'} ~ A^{a a' b_0}_{d d'}
\label{eq:candidate}
\end{eqnarray}
where $c_{a, a',b_0,d,d'} \in \mathbb{C}$ and the canonical basis of the candidate local operator algebra for mode $1$ given by the terms $A^{a a' b_0}_{d d'}$ can be seen in Figure \ref{fig:basis} as planar diagrams.

\begin{figure}[ht]
    \centering
    \includegraphics[width=0.3\textwidth]{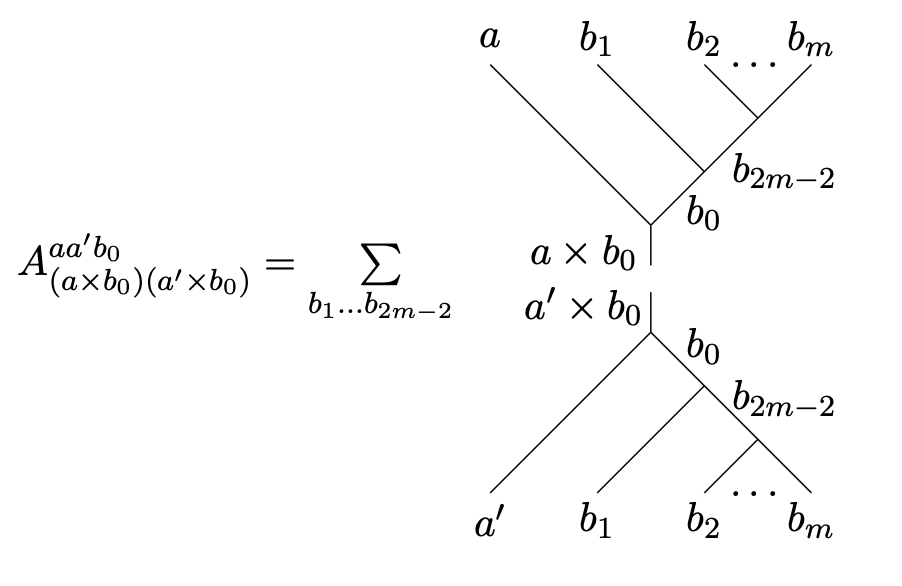}
    \caption{Basis elements of the local operator algebra for the first mode}
    \label{fig:basis}
\end{figure}

Using these basis elements, we want to identify components where the first mode is transformed to the vacuum, as an annihilation operator component would. In anyon diagrams, only one anyon type/charge can be in the same mode. This does not prevent bosonic behaviour, where modes populated with several bosonic particles are expressed as modes being occupied with a higher charge in the anyonic theory. Therefore, the components of the anyonic annihilation operators should consist only of terms that send anyon particle types to the vacuum and not any other particle type. In Figure \ref{fig:alphas}, one can observe that if we fix the particle type $a\neq e$ in mode $1$ bra and the vacuum $e$ in the mode $1$ ket, the basis components then depend only on the global charge of the rest of the system $b_0$ and the term $a\times b_0$, since $e$ is an abelian particle and then $e\times b_0 $ is always $b_0$.

Thus, we realize that the number of annihilation elements that a particle type $a$ has associated in a mode is the number of fusion channels that that particle type has associated with it. This result comes directly from the explicit dependency of having the different annihilating components from $a' \times b_0$, being $b_0$ any particle type.  Thus, all fusion channels of $a'$ will have an associated annihilating element.

For notation, we label each of these annihilation elements of the canonical basis ${a}^{b_0, a\times b_0}_1=A^{e a b_0}_{b_0 a\times b_0}$ (where $1$ expresses the fact they are annihilating on the first mode, $b_0$ and $a\times b_0$ specify the fusion channel and annihilating term, and $a$ is the particle type being annihilated). In all the above and the following expressions, one needs to keep in mind that $a \neq e$.

\begin{figure}[H]
    \centering   \includegraphics[width=0.4\textwidth]{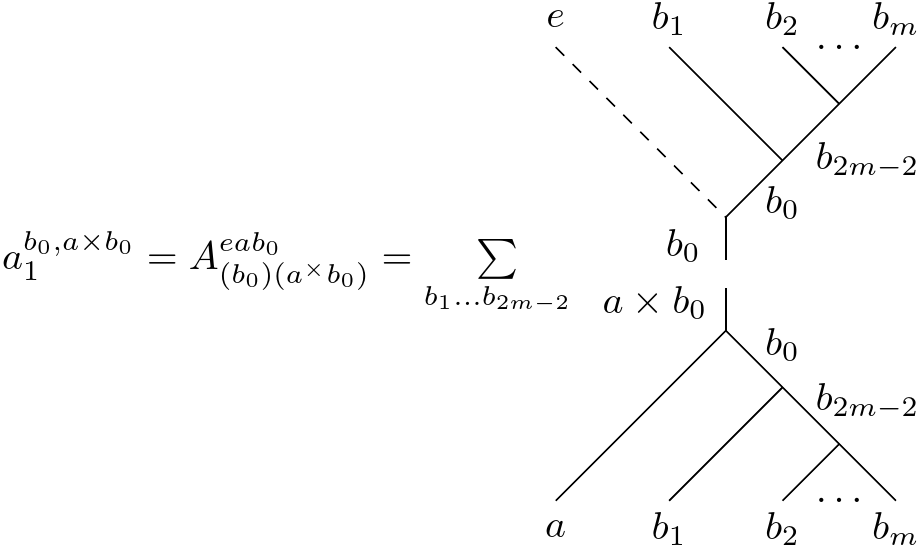}
    \caption{Annihilating elements of the basis of local operators for mode 1. Note that we express the identity anyon with a dashed line.}
    \label{fig:alphas}
\end{figure}

We will refer to the Hermitian conjugate of such annihilating elements as the creating elements. By direct calculation, we find two very exciting results. First is that the annihilating and creating elements of mode $j$ are generators of the candidate local algebra of mode $j$. The second is that the collection of all annihilating and creating elements are generators of the total operator algebra. 

Let us remark on this crucial point. We have seen that the annihilating elements of Figure \ref{fig:alphas}, together with their adjoints, are generators of the candidate local operator algebra. Having obtained these results, we now naturally wonder if the annihilation operators we are looking for are these annihilating elements.

We think they are not. However, we believe that annihilation operators have to be concrete linear combinations of these annihilating elements. In other words, we find that the annihilating elements are components of the annihilation operators, and now we have to decide which is the right way to combine them. 
 
We have these insights by analysing the annihilation operators of spinless fermionic theory in a finite lattice \cite{Friis}. Let us fix the simple setting of having two spinless fermionic modes.

We have a vacuum $\ket{\Omega}$ and two annihilation operators $f_1,f_2$ such that the anticommutation relations hold: 

\begin{eqnarray}
\{f_i,f_j\}=0 \quad  \{f_i,f_j^\dagger\}=\delta_{ij}
\end{eqnarray}

We can represent this theory as an abelian anyon theory with two particle types: a fermion $\psi$ and the vacuum $e$. It is straightforward to see that if we associate each annihilating element with an annihilation operator, we find that instead of a single annihilation operator $f_i$ per mode, we have two annihilation operators per mode: $\psi^{e,\psi}_i$ and $\psi^{\psi,e}_i$ (see Figure \ref{fig:alphas} when replacing $a=\psi$ and summing over the two particle types $e$ and $\psi$). Therefore, this assignment cannot be the correct one. However, we observe that 

\begin{eqnarray}
f_1 = \psi^{e,\psi}_1 + \psi^{\psi,e}_1 \qquad f_2 = \psi^{e,\psi}_2  - \psi^{\psi,e}_2
\label{eq:fermion}
\end{eqnarray}

These relations imply that the fermionic annihilation operators are linear combinations of the annihilation components. In the following lines, we derive which exact linear combinations have to be taken to get the annihilation operators. 

Concretely, we are proposing that the annihilation operators will be operators of the form: 

\begin{eqnarray}
\alpha^{(j)}_k= \sum_{b_0, c_0=a\times b_0} C^{(j)}_{b_0,c_0,k} ~ a_k^{b_0,c_0}
\end{eqnarray}

where $C^{(j)}_{b_0,c_0,k} \in \mathbb{C}$. The term $\alpha$ refers to the fact of being the annihilation operator of the particle type $a$. The label $(j)$ labels the fact that we may need more than one annihilation operator per particle type.

To constraint the coefficients $C^{(j)}_{b_0,c_0,k}$ we consider three conditions that the annihilation operators $\alpha^{(j)}_k$ need to satisfy. The first is that $\{\alpha^{(j)}_{k_1},\dots, \alpha^{(j)}_{k_m} \}_{j,\alpha}$ and their adjoints generate the local algebra of observables in the modes $k_1,\dots,k_m$.

\begin{figure}[H]
    \centering
    \includegraphics[width=0.3\textwidth]{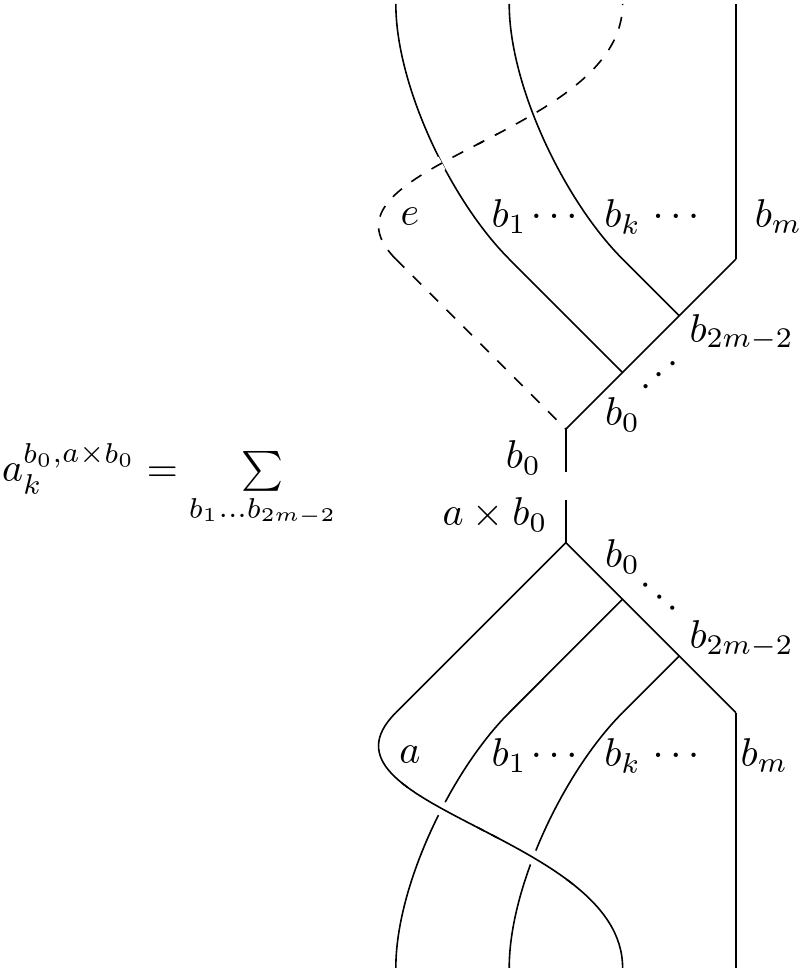}
    \caption{Annihilating elements of the basis of local operators for the $k$th-mode.}
    \label{fig:swap}
\end{figure}

Second, we require that to obtain $\alpha^{(j)}_k$ we only need to know $\alpha^{(j)}_1$ and exchange our way through to $k$. This requirement comes from the intuition that if one wants to annihilate a particle in $k$, it should be equivalent to bringing that particle to $1$, annihilating it there, and then undoing the path we have taken. We show in Figure \ref{fig:swap} that the concrete path we take is the chain of simple exchanges. One could pose different paths giving different annihilation operators, it would be interesting to study the relationship between the definition of subsystems and partial tracing procedures with how this path has to be taken. However, this is further from the scope of this paper. The exchanging condition imposes the following recursive relation to constrain the coefficients $C^{(j)}_{b_0,c_0,k}$

\begin{eqnarray}
\alpha^{(j)}_k= R_{k-1 k} \cdot \alpha^{(j)}_{k-1} \cdot R_{k-1 k}^\dagger
\end{eqnarray}

In the fermionic example that we pose in equation \ref{eq:fermion} we see exactly how the factor $-1$ appears in $f_2$ due to the exchange operation acting on the 'bra' of $\psi^{\psi,e}_1$ non-trivially.

And the third requirement is that for every $b_0, j, k$ there is at least one term $C^{(j)}_{b_0,c_0,k}$ that is non-zero. This is to ensure that the annihilation operators $\alpha^{(j)}_k$ have support on any value of the total charge for the modes other than $k$. This is to prevent explicitly situations where the annihilating terms can be considered annihilation operators and have redundancy.  

We have found a solution to these three constraints. Thus we have found a way to define annihilation operators in anyonic systems. For the solution we propose, the $C_{b_0,c_0,1}^{(j)}\in \mathbb{C}$ we set them to be either $0$ or $1$. However, one could modify our presented solution including different non-zero factors to the terms that are $1$. 

The number of annihilating elements in a mode for the anyon type $a$ is $n_a=\sum_{b c =1}^n N^{c}_{a b}$. Following our general construction, the number of annihilation operators associated with this anyon type $a$ for a given mode will be $J=n_a-n+1$, where $n$ is the total number of particle types in the theory. Notice that with this scheme, we find that for an abelian anyon particle type $a$, there is a single annihilation operator, since for abelian anyon types $n_a=n$ because there are no multiplicities in the fusion channels associated with $a$.

We show how to construct the $J$ annihilation operators for any anyon theory in the Appendix \ref{sec:general}. To make the letter concise, we show here the construction for the simplest non-abelian case, Fibonacci anyons.

We order the Fibonacci particle types as $e,\tau$ of the different allowed fusion channels. We label $c_{b_0,j}$ the $j$'th particle type such that $c_{b_0,j}=\tau \times b_0$. For the first annihilation operator of $\tau$, we set the terms $C^{(0)}_{b_0,c_{b_0,1},1}=1$ and the rest, $C^{(0)}_{b_0,c_{b_0,j},1}$, vanish. This implies that $\alpha^{(0)}_1$ is given by the coefficients being $C^{(0)}_{e,\tau,1}=1$, $C^{(0)}_{\tau,e,1}=1$, and $C^{(0)}_{\tau,\tau,1}=0$.

To define $\alpha^{(1)}_1$, we look at the first $b_0$ with more than one compatible $c_0$. In this case, this is $b_0=\tau$. Now all coefficients remain the same as in $\alpha^{(0)}_1$ except for setting $C^{(1)}_{\tau,c_{b_0,2},1}=1$ and $C^{(1)}_{\tau,c_{b_0,1},1}=0$. Implying that $\alpha^{(1)}_1$ is given by the coefficients being $C^{(1)}_{e,\tau,1}=1$, $C^{(1)}_{\tau,e,1}=0$, and $C^{(1)}_{\tau,\tau,1}=1$.

We would follow the construction to find $\alpha^{(2)}_1$ by applying the same changes but with $c_{\tau,3}$. However, there is no such valid fusion channel. Then we would proceed to the next $b_0$ following the ordering for which $c_{b_0,2}$ exists, and follow the same procedure. In the Fibonacci case, there is no next $b_0$. Thus the construction has been completed. 

We obtain for the Fibonacci case that $\tau$ has $J=2$, annihilation operators. See Figure \ref{fig:fib_op} for a diagrammatic representation of the Fibonacci annihilation operators for a three-anyon Fibonacci space.

Under this general construction that can be found in Appendix \ref{sec:general} and using the simple algebraic identities of the annihilation elements that conform to the annihilation operators, it is straightforward to check that the collection of all annihilation and creation operators for all modes can generate the global algebra of operators and, henceforth, of observables in particular.

Using direct computation, it is also straightforward to check that the annihilation and creation operators for a set of modes generate the local algebra of observables for such a set of modes. The general proof can be found in Appendix \ref{sec:proofobs}. We provide specific examples in the next section expressing Fibonacci observables in terms of the creation and annihilation operators.

\section{Example}\label{sec:examples}
We now exemplify the general results focusing on Fibonacci anyons. Let us start by looking at the generators of the annihilation algebra on three Fibonacci anyons. In Figure \ref{fig:fib_op}, we show the three annihilating elements for a Fibonacci anyon $\tau$ in the left lattice site $1$ and central lattice site $2$. Note that the operators acting on the site $2$, $\tau_2^{b_0, c_0}$, can be obtained from $\tau_1^{b_0, c_0}$ by exchanging the anyons on $1$ and $2$. We express all the operators in the canonical basis by using the F-matrices $F^{abc}_d$ and exchanging factors $R_c^{ab}$ presented at the start of this publication. 

\begin{widetext}

\begin{figure}[H]
    \centering
    \includegraphics[width=0.5\textwidth]{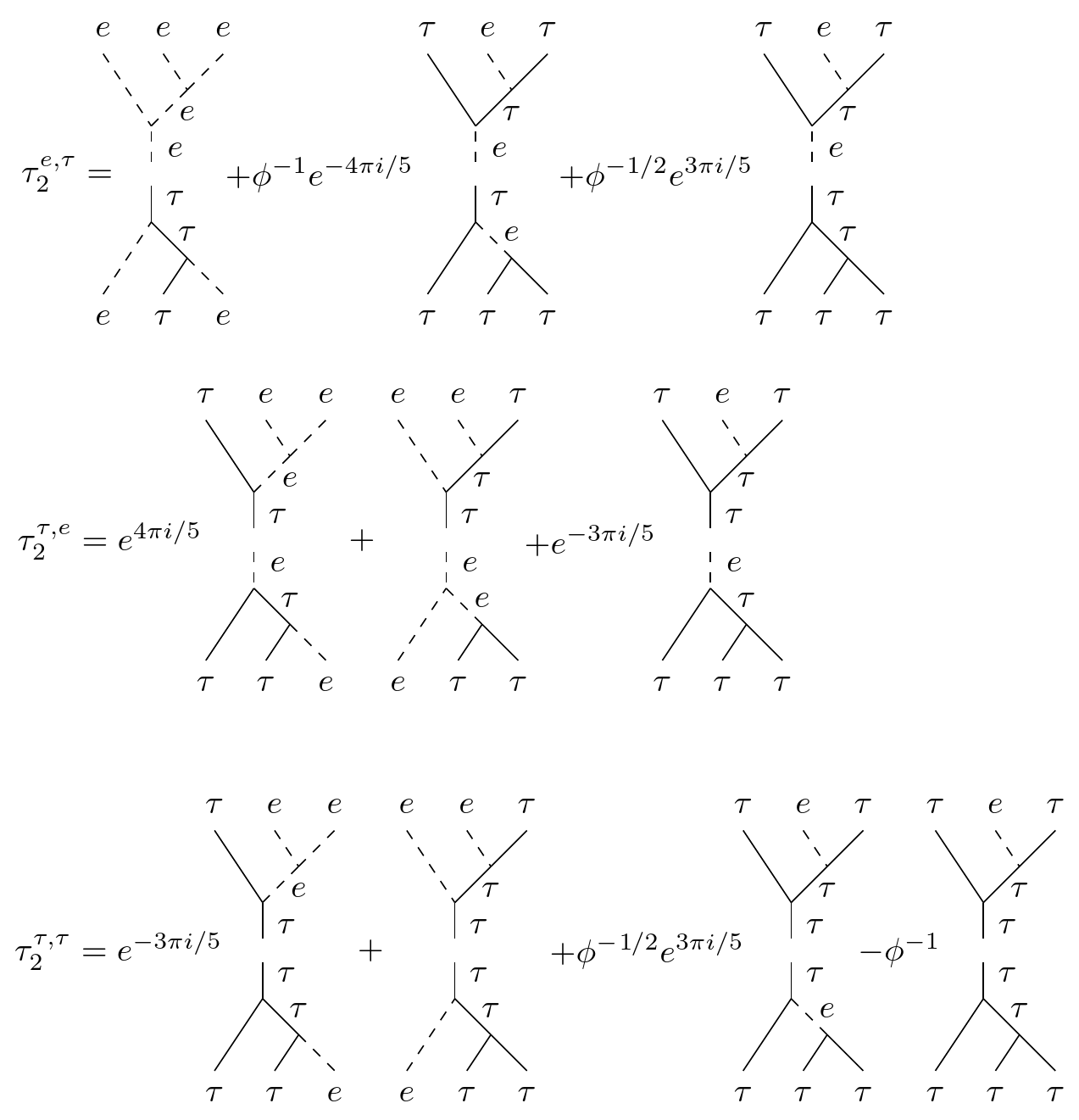}
    \caption{Annihilation elements acting on the first ($\tau^{b_0,c_0}_1$) and second mode ($\tau^{b_0,c_0}_2$) of a three Fibonacci anyon system. The factors are obtained from the change of basis by applying the $F$ and $R$ matrices. $\phi^{-1}$ is the inverse golden ratio, $\phi^{-1}=\left(\sqrt{5}-1\right)/2$.}
    \label{fig:fib_op}
\end{figure}
\end{widetext}

In Fibonacci anyons, we have to define two annihilation operators $\alpha_k^{\left(1\right)}$,$\alpha_k^{\left(0\right)}$ for the Fibonacci $\tau$ particle type. Both operators use the term $\tau_k^{e,\tau}$. To have better algebraic properties, we choose to add a factor of $\frac{1}{\sqrt{2}}$ in front of such terms that will take into account this repetition. We call these two unnormalised annihilation operators: $\alpha_k$ and $\beta_k$, and we will use them throughout the rest of the text.

\begin{align}
    \alpha_k=\frac{1}{\sqrt{2}}\tau_k^{e,\tau}+\tau_k^{\tau, e}, && \beta_k=\frac{1}{\sqrt{2}}\tau_k^{e,\tau}+\tau_k^{\tau, \tau}.
    \label{eq:fib_alphabeta}
\end{align}

In Figure \ref{fig:observables}, we see how some local observables in modes $1$ \& $2$ can be expressed in terms of the local creation and annihilation operators of such modes. A complete list of all observable terms can be found in Appendix \ref{sec:fibobs}.

\begin{figure}[H]
    \centering
    \includegraphics[width=0.20\textwidth]{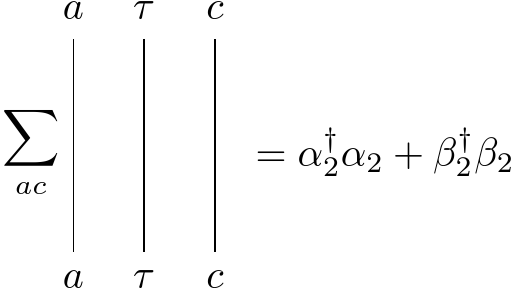}\hfill\includegraphics[width=0.25\textwidth]{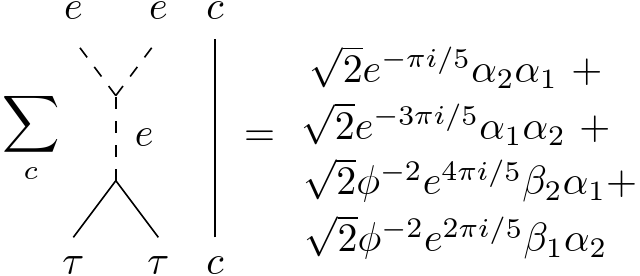}
    \caption{Expression of different Fibonacci observables in terms of anyonic creation and annihilation operators.}
    \label{fig:observables}
\end{figure}

\section{Anyonic Hubbard hamiltonian}\label{sec:hamiltonian}

We want to give use to the annihilation operators that we have defined. A straightforward application is to express Hamiltonians in terms of annihilation operators. Recent work has been studying the properties of Ising-like Fibonacci Hamiltonians \cite{hamiltoniananyon}. 

By expressing Hamiltonians using annihilation operators, we hope to first showcase the similarities and differences between Fibonacci anyons and other particle types such as fermions and bosons; and second, provide tools for the simulation of such Hamiltonian systems, allowing the application of tensor-networks methods \cite{tensornetworks}, explore mapping for applying the Bethe ansatz \cite{Amico} and other methods already used in the 1+1 D case where the notion of annihilation operators is exploited \cite{1dbethe}.

In the Hubbard Hamiltonian described in \cite{hamiltoniananyon} we have a $2 \times N$ square lattice with the ordering shown in Figure \ref{fig:2nordering}. We discuss some consequences and issues that arise from this choice further in the text.

\begin{figure}[H]
    \centering
    \includegraphics[width=0.20\textwidth]{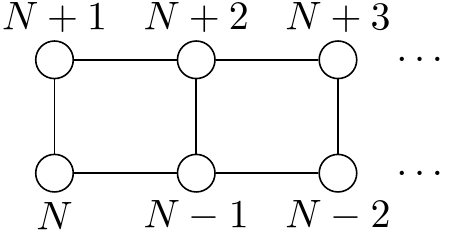}
    \caption{Lattice of model and chosen ordering for a $2\times N$ lattice.}
    \label{fig:2nordering}
\end{figure}

The Hamiltonian has two contributions. First, a hopping contribution between nearest neighbours, where a $\tau$-anyon can jump to the nearest neighbour if it is unoccupied. And a second term, a self-energy term for when there is a $\tau$ in some site. For simplicity and conciseness, we take the same coupling strength for longitudinal and transverse hopping $t_\perp=t_\parallel=t$ \cite{hamiltoniananyon}.    

\begin{figure}[H]
    \centering
    \includegraphics[width=0.45\textwidth]{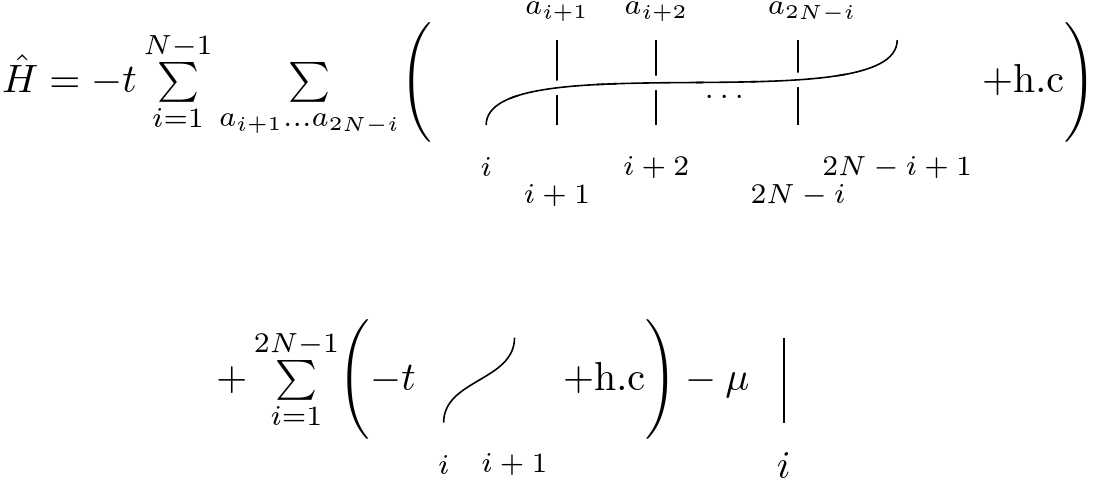}
    \caption{Hubbard Hamiltonian for Fibonacci anyons.}
    \label{fig:hamiltonian}
\end{figure}

The unnormalised annihilation operators $\alpha_k,\beta_k$ allow us to express the Hamiltonian more compactly. It can be expressed without using the unnormalised annihilation operators, but the expression obtained is not as clean and clear as the one obtained using them, which is: 

\begin{gather}
\hat{H}= -t \sum_{i=1}^{N-1}  \left(\alpha_{2N-i+1}^\dagger \alpha_{i} + \beta_{2N-i+1}^\dagger \beta_{i}\right) + \text{h.c}. \nonumber\\ -t \sum_{i=1}^{2N-1}  \left(\alpha_{i+1}^\dagger \alpha_{i} + \beta_{i+1}^\dagger \beta_{i}\right) + \text{h.c}. \nonumber \\ -\mu \sum_{i=1}^{2N} \left(\alpha_i^\dagger \alpha_i + \beta_i^\dagger \beta_i\right)
\label{eq:hamiltonian}
\end{gather} 

We can see how the Hamiltonian has the same terms as in the 2D Fermi-Hubbard model with the same lattice ordering but with two different types of annihilation operators. This expression was not found by directly replacing the fermionic annihilation operators with anyonic annihilation operators. It was found by expressing the Hamiltonian in diagrammatic form in Figure \ref{fig:hamiltonian}, and expressing the diagrammatic observables in terms of the unnormalised anyonic creation and annihilation operators we defined.    

We want to remark that there is nothing in particular of the Hamiltonian in Figure \ref{fig:hamiltonian} which makes it expressable in terms of the creation and annihilation operators. Any physically allowed Hamiltonian can be expressed in terms of the creation and annihilation operators we have defined. It is a matter of convenience to use the unnormalised annihilation operators. These can be described in terms of the original normalised annihilation operators as $\alpha_{j}=\frac{1}{\sqrt{2}}\alpha^{(1)}_{j} {\alpha^{(0)}_{j}}^\dagger \alpha^{(0)}_{j}+\alpha^{(0)}_{j}-\alpha^{(1)}_{j} {\alpha^{(0)}_{j}}^\dagger \alpha^{(0)}_{j}$ and $\beta_{j}=\alpha^{(1)}_{j} {\alpha^{(0)}_{j}}^\dagger \alpha^{(0)}_{j}+ \alpha^{(1)}_{j}-\alpha^{(1)}_{j} {\alpha^{(0)}_{j}}^\dagger \alpha^{(0)}_{j}$.

Nevertheless, there is a subtlety. One needs to pick specific lattices and orderings in order to express the desired notion of locality, as we comment in Figure \ref{fig:planarreps}. That is because we have defined the annihilation operators in the different sites as the annihilation operator in the first site and swapped them in front the other sites. However, we could have defined the annihilation operators at the $k$th site as the annihilation operator in the first site and swapped them in-front the other lattice sites. The resulting two annihilation operators would generate inequivalent spaces because they correspond to two equivalent notions of locality; they are associated with two different subsystem partitions.

In order to express the correct notion of nearest neighbour locality in terms of the annihilation operators we defined (in-front) alone, one needs to pick the ordering such that the connection happens in-front all the modes between the ones in the connection. We want to explore this further in future works and be able to prove the conjecture that for any planar lattice, one can find an ordering such that all the nearest neighbour links can be made to happen either completely behind the in-between modes or completely in front. Thus, making any nearest neighbour Hamiltonian expressable in terms of creation and annihilation operators of the neighbouring terms alone. We have strong indications that such a claim is true. However, for length purposes, we prefer to explain it in a different piece fully.

\section{Discussion}\label{sec:discussion}
One may wonder if the general construction of the proposed annihilation operators applies well to the abelian case. Since for abelian particles, the fusion is deterministic and there is a single possible fusion channel. Applying our method, we recover a single annihilation operator per particle type and lattice site as one would expect.

One can stop and think about how remarkable it is that a single annihilation operator per mode can define bosonic and fermionic systems. Why can a single mathematical object describe the local behaviour of fundamental particles? We have seen that the critical property that allows us to have such a description is their abelian nature.       

For non-abelian particles, we observe that the number of annihilation operators per lattice site is $J=n_a-n+1$, where $n_a$ is the total number of allowed fusion channels associated with that particle type. In the Fibonacci case, for the $\tau$ particle $n_a=3$, we have $\tau\times e=\tau$ and $\tau\times \tau=e+\tau$, and $n=2$ because there are two particle types in Fibonacci anyons, therefore $J=2$.  

We want to notice that the construction is general for any non-abelian anyon theory. We have exemplified it with Fibonacci anyons to be concise. Still, annihilation operators can be defined for Ising anyons \cite{eisert} or any other non-abelian anyon theory one would like to work with. For future work, we would like to explore the connections between the annihilation operators defined using this method for Ising anyons with the annihilation operators one has for Majorana fermions. 

This article presents annihilation operators in the diagrammatic formalism for non-abelian 2D anyons. We want to describe the algebraic properties of the anyonic annihilation and creation operators in commutation-like relations to have a complete algebraic characterization of the anyonic theory and be able to perform manipulations at the annihilation operator level without computing at the diagrammatic level.

We have the suspicion that a complete characterization at the algebraic level might be very challenging. We believe that the algebraic rule for determining whether a combination of creation and annihilation operators is superselection-respecting or not might be rather cumbersome. See Appendix \ref{sec:commrels} for some known algebraic relations of the Fibonacci creation and annihilation operators.


If we refer to the fusion tree where all the components are the identity particle type as $\ket{0}$, we see that $\alpha_k^{(j)} \ket{0}=0$ for all $j$ and $k$. It is straightforward to see that $\ket{0}$ is unique under this property. We can now express any state of the canonical basis as a well-ordered sequence of creation operators acting on $\ket{0}$. Concrete expressions for three-mode Fibonacci anyons can be found in Appendix \ref{sec:fock}. One could try to use these expressions to find suitable Jordan-Wigner mappings for 2+1 D anyons.

Furthermore, exploring Bogoliubov-like transformations for the non-abelian anyonic annihilation operators would be interesting. It would allow one to define a different notion of anyonic mode, not tied to the position latticing of the system we introduced.   


We can now describe general anyonic Hamiltonians in $2\times N$ lattices, by using the in-front-only annihilation operators. The ability to express the Hubbard-like anyonic Hamiltonian in terms of local annihilation operators may have implications in the simulation of the model. Until recently, the community was lacking good numerical techniques to simulate non-Abelian anyons systems. The main difficulty comes from the lack of a tensor product structure and the growth of the Hilbert space with the number of particles. There have been some recent efforts to generalize the tensor network formalism to lattice systems of anyons \cite{Singh2014, Ayeni2016, Pfeifer2015}. However, this work defines the anyonic local operators that constitute the Hamiltonian with their crude representation in the diagrammatic formalism. We expect that having access to the local annihilation operators of an anyonic theory will facilitate the numerical simulation in some cases. In this way, we can exploit the parallelism between the anyonic Hubbard Hamiltonian and its bosonic or fermionic counterpart, for instance.

We note that the Hamiltonian in Equation \ref{eq:hamiltonian} has terms with long-range interactions. These highly long-range terms (with respect to the ordering) can make the simulations time-inefficient. The ordering is deliberately chosen to be the one in Figure \ref{fig:2nordering}, so we just need the in-front-only annihilation operators. Because we have defined them with the exchanging going in one direction, we are restricting ourselves to a not-that-simple expression of the exchanging operator in the other direction. Therefore, we choose the lattice ordering so we do not encounter crossings of anyon lines in this direction. Of course, we can think of a more natural (and short-range) ordering, e.g., ladder ordering, but then our Hamiltonian terms will contain products of several local operators in not only the nearest-neighbor interacting sites. However, this non-locality can be avoided by defining two more sets of creation and annihilation operators analogous to the ones defined in this paper but changing the direction of the exchanging.

In conclusion, if we want to avoid long-range terms (with respect to the ordering) in our Hamiltonian we need to sacrifice the simplicity of the current expression. We think that the study of the similarities and differences between these three approaches is a promising future direction to follow.

We hope that having found expressions for the 2+1 D non-abelian anyon creation and annihilation operators will advance the study and understanding of this topic, especially by allowing us to apply known techniques to the study of topological quantum computing and the experimental detection of such particles described.     

\section*{Acknowledgements}\label{sec:aknows}
We want to thank Steve Simon and Babatunde Ayeni for their comments that improved the first draft. We want to thank Steve Simon again for the course on Topological Quantum Matter at the University of Oxford and the exquisitely well-written course Lecture Notes. Lucia Vilchez-Estevez thanks the Clarendon Foundation for providing financial support during the development of this project. This work was suported by the John Templeton Foundation through the ID\# 62312 grant, as part of the 'The Quantum Information Structure of Spacetime' Project (QISS). The opinions expressed in this publication are those of the authors and do not necessarily reflect the views of the John Templeton Foundation. 

\bibliographystyle{apsrev4-1}
\bibliography{references}
\appendix

\section{General annihilation operators}
\label{sec:general}

We show how to construct the $J$ annihilation operators for any anyon theory. As we expose in the main text, we have identified the annihilating elements for the first mode $a_{1}^{b_0, a\times b_0}$ (see Figure \ref{fig:alphas}). We have also seen that for a general mode $k$, we define the annihilating elements according to the notion of mode locality where we exchange the first mode in-front all the $k-1$ others until $k$ (see Figure \ref{fig:swap}).



Now, to make the annihilation operators, we have seen in the text that as an analogy to the known fermionic annihilation operators, we need to take linear combinations of the annihilating elements. This is without messing with the properties of spanning the local algebra of observables. We want to construct the normalized annihilation operators by specifying the coefficients $C^{(j)}_{b_0,c_0,k} \in \mathbb{C}$ of the linear combinations of the annihilating terms: 

\begin{eqnarray}
\alpha^{(j)}_k= \sum_{b_0, c_0=a\times b_0} C^{(j)}_{b_0,c_0,k} ~ a_k^{b_0,c_0}
\end{eqnarray} 

As we explain in the text, under exchanging in-front the other modes, $C^{(j)}_{b_0,c_0,k}$ is determined if one knows the $C^{(j)}_{b_0,c_0,1}$. Without loss of generality, we are interested in the normalised annihilation operators where $C^{(j)}_{b_0,c_0,1}$ is either $0$ or $1$; the following arguments can be repeated in general because the relevance is on which $C^{(j)}_{b_0,c_0,1}$ need to vanish.   

If we try to have a single annihilation operator, for $\alpha^{(0)}_1$ to be able to generate the whole algebra of observables, it is necessary that no coefficient vanishes. We see here a big difference between abelian and non-abelian particles. For particles that are abelian (there is a single possible value for $c_0=a\times b_0$ for any $b_0$), it is straightforward to see that all coefficients $|C^{(0)}_{b_0,c_0,1}|=1$ then $\alpha^{(0)}_1{\alpha^{(0)}_1}^\dagger$ and ${\alpha^{(0)}_1}^\dagger\alpha^{(0)}_1$ generate the local algebra of observables for mode $1$. 

However, considering that the particle $a$ is not abelian, then there exist two fusion channels compatible for $a\times b_0=c_0+c'_0$ for some $b_0$. Now when taking ${\alpha^{(0)}_1}^\dagger\alpha^{(0)}_1$, terms that violate the superselection rule appear. Specifically, terms that convert total charge $c_0$ to total charge $c'_0$ and vice versa. It is straightforward to observe that one cannot get rid of these terms while keeping the relevant terms necessary to have the local observable that is the projector of the $a$ particle type in mode $1$ by adding more terms in the monomial. The only way to make these undesired terms go away is by either setting $C^{(0)}_{b_0,c_0,1}=0$ or $C^{(0)}_{b_0,c'_0,1}=0$. Let us consider, without loss of generality; we have done the second. We now need at least another annihilation operator with a vanishing $C^{(j)}_{b_0,c_0,1}$ and a non-vanishing $C^{(j)}_{b_0,c'_0,1}$ such that the terms with total charge $c'_0$ that appear in the projector of particle type $a$ in mode $1$ can be generated. 

We have seen that for each "extra" fusion channel, we need at least one annihilation operator that contains such a term. The general construction we provide guarantees such property. However, it might not be optimal. There may be a different grouping of the terms such that the total number of annihilation operators per non-abelian particle is smaller. We know that at least the number of annihilation operators needs to be $J'=\max_k (\sum_{l} N^{a_l}_{a a_k})$. In the construction we present, we obtain $J=\sum_{k,l} N^{a_l}_{a_j a_k}-n+1$ annihilation operators. For the two most relevant non-abelian anyon families, Fibonacci and Ising anyons, we have that $J'=J$; therefore, our construction is optimal for these important cases. 

The construction is as follows. First, we fix an order in the particle types of the anyon theory. We choose to bring all abelian particles at the beginning of the ordering. This fixed order defines a preferred basis for the $n\times n $ matrices for each particle type $a_j$ defined as $(A_j)_{k l}=N^{a_l}_{a_j a_k}$. Now to construct the $J=\sum_{k,l} N^{a_l}_{a_j a_k}-n+1$ annihilation operators for $a_j$ we  label $c_{a_k,i}$ the $i$'th particle type such that $c_{a_l,k}=a_j \times a_k$. For the first annihilation operator of $a_j$, we set the terms $C^{(0)}_{a_k,c_{a_k,1},1}=1$ and the rest, $C^{(0)}_{a_k,c_{a_k,i},1}$ for $i>1$, vanish. This is analogous to selecting the first fusion channel in each row of $A_j$. The choice of the coefficients defines the first annihilation operator. 

If $a_j$ is abelian, our work is over since then $J=1$, and there are no $C^{(0)}_{a_k,c_{a_k,i},1}$ for $i>1$. However, if $a_j$ is non-abelian we need to construct the other annihilation operators. To do so, we go to the first $a_{k_0}$ in the ordering such that exists a $c_{a_{k_0},2}$; so, that has more than one allowed fusion channel when fusing $a_j$ with $a_{k_0}$. Once identified, we set $C^{(1)}_{a_k,c_{a_k,1},1}=1$ for $k\neq k_0$, $C^{(1)}_{a_{k_0},c_{a_{k_0},2},1}=1$, and making all others coefficients vanish. This specification would fix the second annihilation operator.  

To produce a third annihilation operator for $a_j$, we would first check if there exists a $c_{a_{k_0},3}$, if it does we would set $C^{(2)}_{a_k,c_{a_k,1},1}=1$ for $k\neq k_0$, $C^{(2)}_{a_{k_0},c_{a_{k_0},3},1}=1$, and making all others coefficients vanish. Thus setting the third annihilation operator. If $c_{a_{k_0},3}$ does not exist, though, we then go to the next $a_{k_1}$ in the ordering such that $c_{a_{k_1},2}$ exists, and we would set $C^{(2)}_{a_k,c_{a_k,1},1}=1$ for $k\neq k_1$, $C^{(2)}_{a_{k_1},c_{a_{k_1},2},1}=1$, and making all others coefficients vanish. Therefore we would have specified the third annihilation operator. 

To produce the $m$th annihilation operator, one can see the recursive strategy we are following. Given the special $a_{k'}$ we have identified in the $m-1$th annihilation operator for which we have set $C^{(m-1)}_{a_{k},c_{a_{k'},i},1}=1$ for $i>1$, then either $c_{a_{k'},i+1}$ exists or not. If it does we set $C^{(m)}_{a_k,c_{a_k,1},1}=1$ for $k\neq k'$, $C^{(m)}_{a_{k'},c_{a_{k'},i+1},1}=1$, and make all others coefficients zero. However, if it does not exist, we find the next $a_{k''}$ from $a_{k'}$ in the ordering such that $c_{a_{k''},2}$ exists. Then, we specify $C^{(m)}_{a_k,c_{a_k,1},1}=1$ for $k\neq k''$, $C^{(m)}_{a_{k''},c_{a_{k''},2},1}=1$, and make all others coefficients zero. 

The process terminates when $c_{a_{k'},i+1}$ does not exist and there is no $a_{k''}$ further down the ordering than $a_{k'}$ such that $c_{a_{k''},2}$ exists. It is straightforward to check that this will happen for the $J$th term.  

The above procedure fixes the annihilation operators for each particle type $a_j$ in the first lattice site. For the other lattice sites, we exchange the annihilation operators in position, bringing the first lattice site in-front the others into the $k$th lattice site.

We would like just to make a technical remark where the identity particle $e$ can also be considered to have an annihilation operator per mode (being an abelian particle). However, the identity annihilation operator can always be expressed in terms of the other annihilation operators of the theory. Concretely, any ${\left(\alpha_l\right)}^{(j)}_{k} {{\left(\alpha_l\right)}^{(j)}}_{k}^\dagger$ will give such annihilation operator. 

\section{Creation and annihilation operators theorem}
\label{sec:proofobs}
Using the construction of the anyonic annihilation operators specified in Appendix \ref{sec:general}, we can prove the desirable properties of the anyonic annihilation operators. In particular, we want to ensure that with this construction, it is possible to express any local observable in a set of lattice sites in terms of the creation and annihilation operators of such lattice sites.  

Concretely the theorem we prove is the following. 

\begin{theorem}
Consider a general anyon theory with $n$ particle types and $N$ lattice sites. Consider a set of lattice sites $\{s_1,,\dots,s_M\}$ and the subsystem bipartition where the selected sites are always in-front the other $N-M$ sites. Under this bipartition, any local observable in these $M$ sites can be written as a polynomial of the creation and annihilation operators of these lattice sites.   
\label{thm}
\end{theorem}

\begin{proof}
To prove this general theorem, we first prove that the statement is true for the sets of lattice sites $\{1,\dots,M\}$, and then we prove that that implies the statement holds for any set of lattice sites.

\begin{figure}[H]
    \centering
    \includegraphics[width=0.85\linewidth]{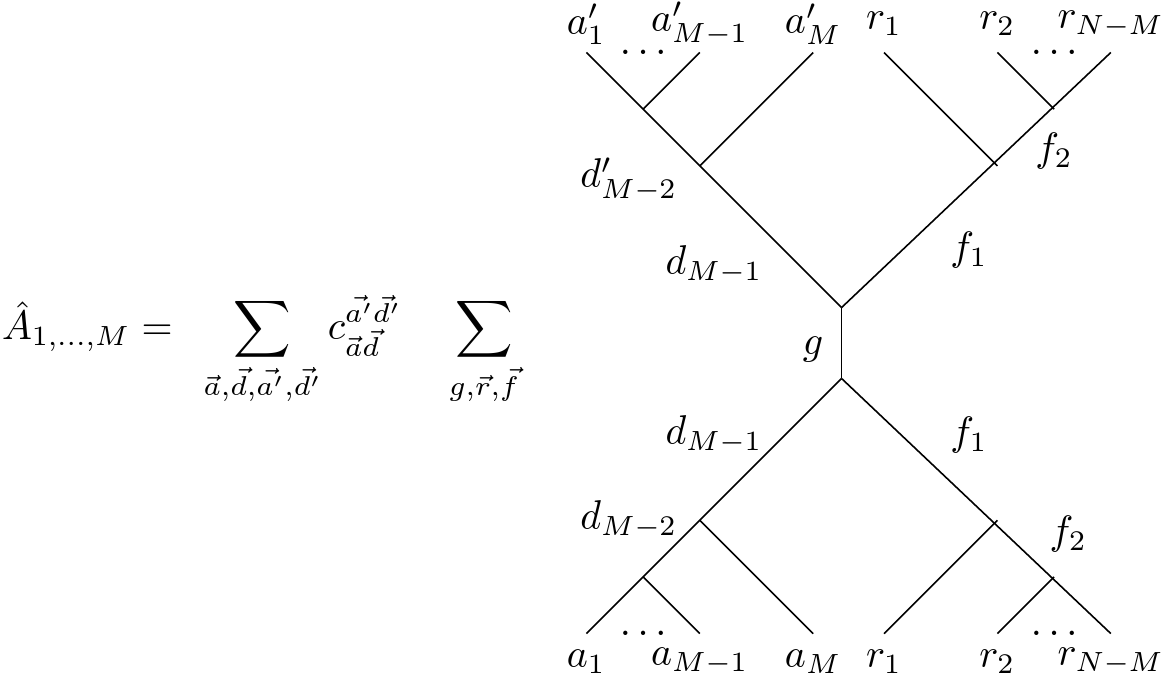}
    \caption{Local observables in $1, \dots, M$. With $c_{\vec{a},\vec{d}}^{\vec{a'},\vec{d'}}={c_{\vec{a'},\vec{d'}}^{\vec{a},\vec{d}}}^*\in \mathbb{C}$.}
    \label{fig:localobs}
\end{figure}

Given the set of lattice sites $\{1,\dots, M\}$, a local observable for the chosen bipartition has the general form shown in Figure \ref{fig:localobs}. Under this bipartition of $1 \dots M | M+1 \dots N$, we can find the elements of the candidate local algebra of operators following the same procedure as in the main paper. We see that any local observable will be an element of the candidate local algebra of operators since it is left invariant by local unitaries acting on the complement of $1,\dots, M$. We define the operators $\hat{O}_{\vec{a},\vec{d},g}$ as in Figure \ref{fig:basisalg}. Notice that any local observable in $1,\dots, M$ can be written as a linear combination of $\hat{O}_{\vec{a},\vec{d},g}^\dagger \hat{O}_{\vec{a'},\vec{a'},g}$.    

\begin{figure}[H]
    \centering
    \includegraphics[width=0.6\linewidth]{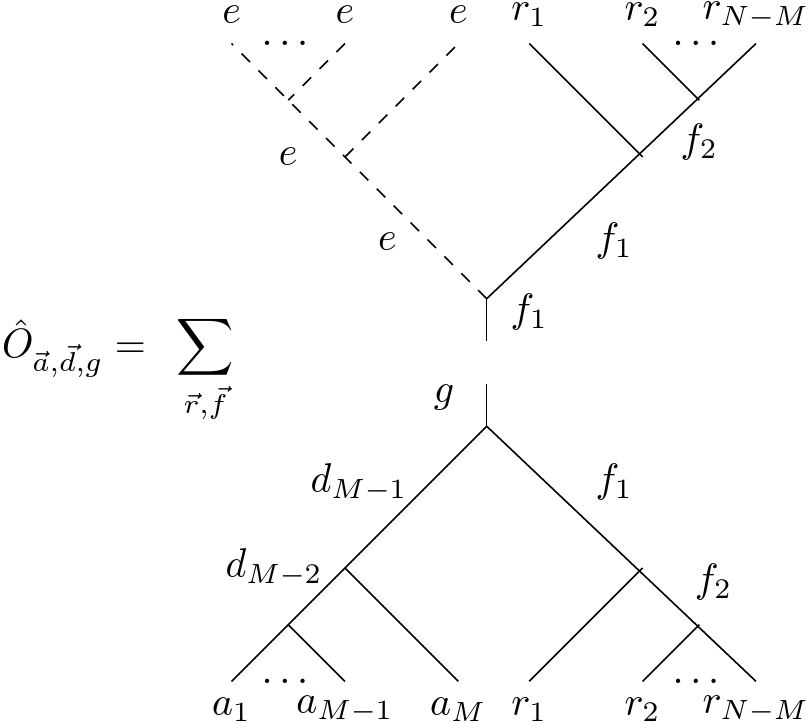}
    \caption{Operators that generate the local observables in $1, \dots, M$}
    \label{fig:basisalg}
\end{figure}

If we can express such operators $\hat{O}_{\vec{a},\vec{d}}$ as polynomials of the local creation and annihilation operators of the lattice sites $1,\dots, M$ then we can express any local observable in these lattice sites as polynomials of the local creation and annihilation operators of the lattice sites $1,\dots, M$.

Consider the annihilation operators ${\left(\alpha_l\right)}^{(j)}_k$ for one of the non-abelian particle types $a_l$. Notice that due to the particle type being non-abelian, there is strictly more than one annihilation operator associated with this particle type. Note that we can compute ${{\left(\alpha_l\right)}^{(j)}}_k-{{\left(\alpha_l\right)}^{(j)}}_k{{\left(\alpha_l\right)}^{(0)}}_k^\dagger {{\left(\alpha_l\right)}^{(0)}}_k= {\left(a_l\right)}_k^{b_j, c_j}$, where $c_j=a_l\times b_j$ such that is not the first in the ordering for the fusion of $a_l$ and $b_j$. With this calculation, we see that we can retrieve from polynomials of the creation and annihilation operators all the annihilating terms that do not appear in ${{\left(\alpha_l\right)}^{(0)}}_k$. Moreover, by now calculating ${{\left(\alpha_l\right)}^{(0)}}_k-{{\left(\alpha_l\right)}^{(j)}}_k+{\left(a_l\right)}_k^{b_j, c_j}={\left(a_l\right)}_k^{b_j, c_0}$ where $c_0=a_l\times b_j$ is the first allowed fusion channel between $a_l$ and $b_j$ under the fixed ordering. We can also calculate ${{\left(\alpha_l\right)}^{(0)}}_k-\sum_{b_j} {\left(a_l\right)}_k^{b_j, c_0}=\sum_{b_r} {\left(a_l\right)}_k^{b_r, a_l\times b_r}$ where the sum over $b_j$ is over the particle types $b_j$ that have more than one allowed fusion channel with $a_l$, and the sum over $b_r$ is over the particle types $b_r$ that have only one allowed fusion channel with $a_l$: $a_l\times b_r$. 

For every $b_r$ that is not an abelian particle, we have that there will exist some particle type $a_{s}$ such that there is more than one allowed fusion channel $c_t=a_s \times b_r$. Thus, the terms ${\left(a_s\right)}_k^{b_r, c_t}$ can be expressed as polynomials of the creation annihilation operators for the particle type $a_s$ as we have shown before. It is easy to see that ${\left(a_s\right)}_k^{b_{r}, c_t} {{\left(a_s\right)}_k^{b_r, c_t}}^\dagger \sum_{b_{r'}} {\left(a_l\right)}_k^{b_{r'}, a_l\times b_{r'}}={\left(a_l\right)}_k^{b_{r}, a_l\times b_r}$. 

After all these calculations, we can conclude that we can express any annihilating term $\left(a_j\right)_k^{b_0,c_0}$ in terms of local creation and annihilation operators on $k$, for $b_0$ being a non-abelian particle type. For the abelian terms, we know we can express $\sum_{b} \left(a_l\right)_k^{b,b\times a_l}$ where the sum runs over $b$ being abelian particle types, in terms of the creation and annihilation operators. 

Once we have these results, we are ready to see how we can express $\hat{O}_{\vec{a},\vec{d},g}$ in terms of $\left(a_j\right)_k^{b_0,c_0}$ and $\sum_{b} \left(a_l\right)_k^{b,b\times a_l}$, for $b_0$ non-abelian and $b$ abelian and $k\leq M$. Thus, we would see that we can express $\hat{O}_{\vec{a},\vec{d},g}$ in terms of the local creation and annihilation operators in the modes $1, \dots, M$. 

We can see with direct computation the following expression:
\begin{widetext}
\begin{eqnarray}
    \hat{O}_{\vec{a},\vec{d},g}=\prod_{j=2}^{M} \left(\sum_{\substack{b_{M-j+2} \\c_{M-j+2}}} \left[F^{d_{M-j} a_{M-j+2} b_{M-j+2}}_{g}\right]^*_{d_{M-j+1} c_{M-j+2}} {(a_{M-j+2})}^{b_{M-j+2},c_{M-j+2}}_{M-j+2}\right) \cdot \sum_{b_1} {(a_{1})}^{b_{1},g}_{1}
    \label{eq:mastereq}
\end{eqnarray}
\end{widetext}

where $d_0=a_1$. We can see that we can express each term of the product in terms of the local creation and annihilation operators. Let us start with the term $\sum_{b_1} {(a_{1})_1}^{b_{1},g}$ we can decompose the sum between the sum over abelian particles plus the sum over non-abelian particles. Each term of the non-abelian sum can be expressed in terms of the creation and annihilation operators in mode $1$, so the sum of such terms is also. Moreover, we have seen how the sum over the abelian terms is expressible in terms of the creation and annihilation operators. 

Similarly, for the terms in the product involving the $F$-matrices components, we do the same decomposition of the sum. For the non-abelian particles $b_{M-j+2}$, each term ${(a_{M-j+2})}^{b_{M-j+2},c_{M-j+2}}_{M-j+2}$ is expressable in terms of the creation and annihilation operators, thus the linear combinations of such terms can be expressed in terms of such local operators. For abelian particles, it should be noted that if $b_{M-j+2}$ is abelian, $\displaystyle[F^{d_{M-j} a_{M-j+2} b_{M-j+2}}_{g}]^*_{d_{M-j+1} c_{M-j+2}}$ equals $\delta_{c_{M-j+2}, a_{M-j+2}\times b_{M-j+2}}$. This follows from the fact that the F-moves for abelian particles are trivial. Therefore the sum over the abelian particles ends up becoming $\displaystyle\sum_{b_{M-j+2}} {(a_{M-j+2})}^{b_{M-j+2},a_{M-j+2}\times b_{M-j+2}}_{M-j+2}$. Thus, expressable in terms of the local creation and annihilation operators of the mode $M-j+2$.

This concludes that $\hat{O}_{\vec{a},\vec{d},g}$ can be expressed in terms of the creation and annihilation operators. Therefore, any local observable in $1,\dots, M$ can be expressed using the $1, \dots, M$ creation and annihilation operators. Concretely, as a polynomial of such operators. Moreover, note that our proof is constructive and that by using it, one could find a closed expression of any observable in terms of our local creation and annihilation operators.

All we have left now to prove the general theorem is to use the fact that we know the theorem holds for the set of modes $1,\dots, M$ to extend it to any set of modes $s_1,\dots, s_M$ ($s_i<s_{i+1}$). Notice that $i\leq s_i$ always. We will see that we can find a good map between the local observables in a general set of modes $s_1,\dots, s_M$ and the local observables $1,\dots, M$. 

Remember that the notion of locality is such that the relevant modes go in-front the ancillary modes. It is straightforward to observe that applying the unitary transformation $U=\prod_{i=0}^{M\shortminus1}\prod_{j=M\shortminus i+1}^{s_{M\shortminus i}} R^\dagger_{j\shortminus 1 ~ j}$ we transform any local observable in the modes $s_1,\dots, s_M$ to a local observable in $1\dots M$.  
\begin{eqnarray}
    U \hat{A}_{s_1,\dots,s_M} U^\dagger=\hat{A}_{1,\dots ,M}
\end{eqnarray}

Remember that we proved that $\hat{A}_{1,\dots ,M}=p({\left(\alpha_l\right)}^{(j)}_{k},{\left(\alpha_l\right)^{(j')}}_{k'}^\dagger)$, where $p(\cdot)$ is a polynomial, and $k,k'\in \{1,\dots,M\}$. Hence, by linearity and unitarity, $\hat{A}_{s_1,\dots,s_M}=U^\dagger \hat{A}_{1,\dots ,M} U$ equals $p(U^\dagger {\left(\alpha_l\right)}^{(j)}_{k} U,U^\dagger {\left(\alpha_l\right)^{(j')}}_{k'}^\dagger U)$. Notice that $U^\dagger=\prod_{i=1}^M \prod_{j=i}^{s_i\shortminus 1} R_{s_i\shortminus 1\shortminus j+i ~ s_i\shortminus j+i}$,  that the annihilation operators in $k$ are left invariant by unitaries local in the set of modes that excludes $k$, and that ${\left(\alpha_l\right)^{(j)}}_{k}= V {\left(\alpha_l\right)}^{(j)}_{1} V^\dagger$ by definition, where $V=\prod_{j=0}^{k\shortminus 2} R_{k\shortminus 1\shortminus j ~ k\shortminus j}$. Now, it is easy to see that 
\begin{widetext}
\begin{eqnarray}
    U^\dagger {\left(\alpha_l\right)}^{(j)}_{k} U = \prod_{i=1}^k \prod_{j=i}^{s_i\shortminus 1} R_{s_i\shortminus 1\shortminus j+i ~ s_i\shortminus j+i} \prod_{i=k+1}^{M} \prod_{j=i}^{s_i\shortminus 1} R_{s_i\shortminus 1\shortminus j+i ~ s_i\shortminus j+i} \cdot ~\left({\left(\alpha_l\right)}^{(j)}_{k}\right)\cdot \nonumber \\ \cdot \left(\prod_{i=k+1}^{M} \prod_{j=i}^{s_i\shortminus 1} R_{s_i\shortminus 1\shortminus j+i ~ s_i\shortminus j+i}\right)^\dagger \left(\prod_{i=1}^k \prod_{j=i}^{s_i\shortminus 1} R_{s_i\shortminus 1\shortminus j+i ~ s_i\shortminus j+i}\right)^\dagger,
\end{eqnarray}

Moreover, the unitary $\prod_{i=k+1}^{M} \prod_{j=i}^{s_i\shortminus 1} R_{s_i\shortminus 1\shortminus j+i ~ s_i\shortminus j+i}$ is local on the set of modes that do not include $k$, so it leaves any creation operator invariant. Thus, 

\begin{eqnarray} 
U^\dagger {\left(\alpha_l\right)}^{(j)}_{k} U=\prod_{i=1}^{k\shortminus 1} \prod_{j=i}^{s_i\shortminus 1} R_{s_i\shortminus 1\shortminus j+i ~ s_i\shortminus j+i}  \left(\prod_{j=k}^{s_k\shortminus 1} R_{s_k\shortminus 1\shortminus j+k ~ s_k\shortminus j+k}\right) \cdot ~\left({\left(\alpha_l\right)}^{(j)}_{k}\right) ~ \cdot \nonumber \\ \cdot  \left(\prod_{j=k}^{s_k\shortminus 1} R_{s_k\shortminus 1\shortminus j+k ~ s_k\shortminus j+k}\right)^\dagger \left(\prod_{i=1}^{k\shortminus 1} \prod_{j=i}^{s_i\shortminus 1} R_{s_i\shortminus 1\shortminus j+i ~ s_i\shortminus j+i}\right)^\dagger
\end{eqnarray}

By the definition of ${\left(\alpha_l\right)}^{(j)}_{k}$, we see that applying the unitary action of $\left(\prod_{j=k}^{s_k\shortminus 1} R_{s_k\shortminus 1\shortminus j+k ~ s_k\shortminus j+k}\right)$ to it it gives us  ${\left(\alpha_l\right)}^{(j)}_{s_k}$, giving:

\begin{eqnarray} 
U^\dagger {\left(\alpha_l\right)}^{(j)}_{k} U=\prod_{i=1}^{k\shortminus 1} \prod_{j=i}^{s_i\shortminus 1} R_{s_i\shortminus 1\shortminus j+i ~ s_i\shortminus j+i} \left({\left(\alpha_l\right)}^{(j)}_{s_k}\right)  \left(\prod_{i=1}^{k\shortminus 1} \prod_{j=i}^{s_i\shortminus 1} R_{s_i\shortminus 1\shortminus j+i ~ s_i\shortminus j+i}\right)^\dagger
\end{eqnarray}
\end{widetext}

The only thing that is just left to check is that $\left(\prod_{i=1}^{k\shortminus 1} \prod_{j=i}^{s_i\shortminus 1} R_{s_i\shortminus 1\shortminus j+i ~ s_i\shortminus j+i}\right)$ is local on the set of modes that does not include $s_k$. We observe that, indeed, the largest mode that appears in the expression is $s_{k\shortminus 1}$, which we know is strictly smaller than $s_k$. Therefore,  all modes that appear in the expression are strictly smaller than $s_k$, making the unitary local on the set of modes that excludes $s_k$, thus the unitary action leaves $\left({\left(\alpha_l\right)}^{(j)}_{s_k}\right)$ invariant, giving $U^\dagger {\left(\alpha_l\right)}^{(j)}_{k} U={\left(\alpha_l\right)}^{(j)}_{s_k}$. This also proves it for the creation operators $U^\dagger {{\left(\alpha_l\right)}^{(j')}}^\dagger_{k'} U={{\left(\alpha_l\right)}^{(j')}}_{s_k'}^\dagger$. Therefore, we indeed see that any local observable $\hat{A}_{s_1,\dots,s_M}$ in any set of modes $s_1, \dots, s_M$ can be written as a polynomial of the local creation and annihilation operators for such set of modes, since: $\hat{A}_{s_1,\dots,s_M}=\hat{A}_{s_1,\dots,s_M}=U^\dagger \hat{A}_{1,\dots ,M} U= p(U^\dagger {\left(\alpha_l\right)}^{(j)}_{k} U, U^\dagger {\left(\alpha_l\right)^{(j')}}_{k'}^\dagger U)=p({\left(\alpha_l\right)}^{(j)}_{s_k}, {\left(\alpha_l\right)^{(j')}}_{s_k'}^\dagger)$. This concludes the proof of the theorem. 
\end{proof}

\section{3-anyon Fibonacci observables}
\label{sec:fibobs}

We present a complete list of all observable terms local in the $1,2$ modes of a three-mode Fibonacci anyons model. Up to hermitian conjugation, there are nine linearly independent terms. In Figures \ref{fig:locale} \& \ref{fig:localtau}, we show all of them.

Notice that the expressions we show in Figures \ref{fig:locale} \& \ref{fig:localtau} are more compressed than the expressions obtained through the reasoning in the proof of the general theorem. We present these expressions because we think it is more convenient to work with them, especially when investigating Hamiltonians.  

\begin{figure}[H]
\centering
         \hfil\includegraphics[width=0.19\textwidth]{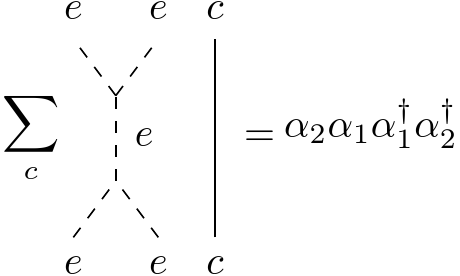} \hspace{15pt}
         \includegraphics[width=0.24\textwidth]{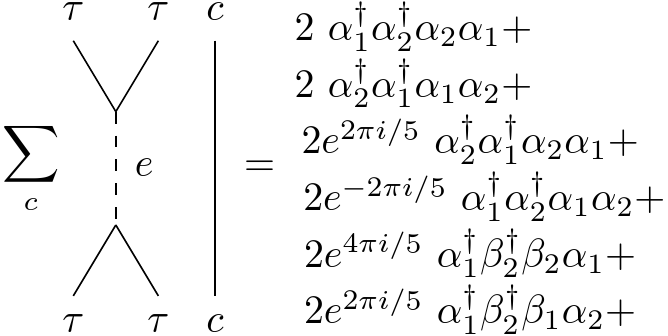}\hfil
     \newline
     \newline
     \hfil \includegraphics[width=0.27\textwidth]{diagrams/observables2new.png}
     \hfil 
    \caption{Local observable terms in $1,2$ with global charge $e$, up to hermitian conjugation}
    \label{fig:locale}
\end{figure}

\begin{figure}[ht] 
        \centering\includegraphics[width=0.27\textwidth]{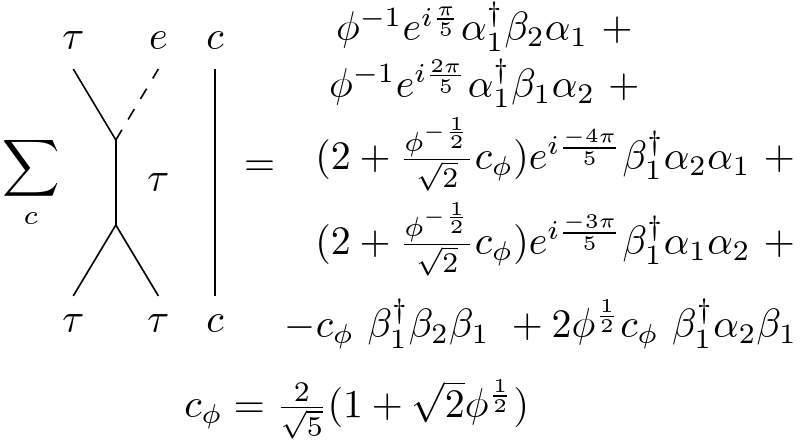}\hfill\includegraphics[width=0.17\textwidth]{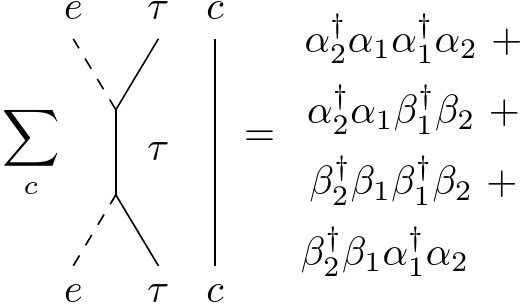}
         \newline
         \newline\includegraphics[width=0.27\textwidth]{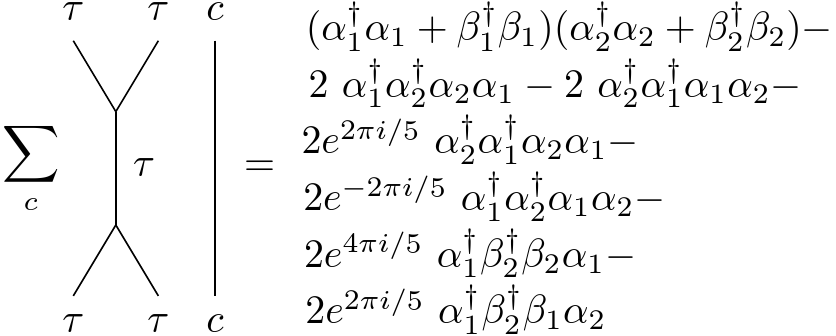}\hfill\includegraphics[width=0.17\textwidth]{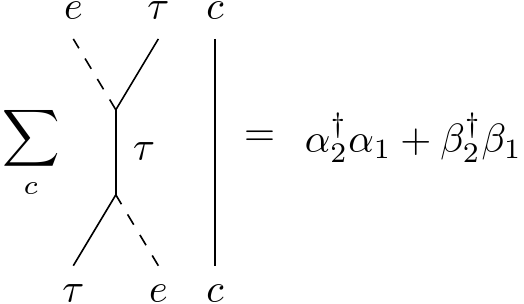}
         \newline 
         \newline\includegraphics[width=0.27\textwidth]{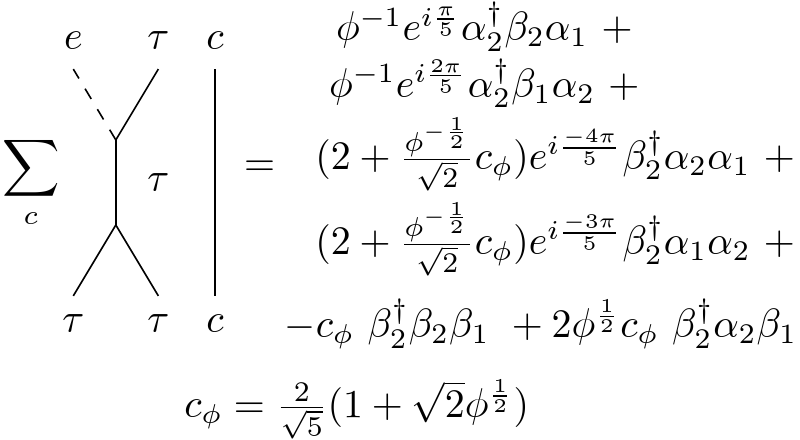} \hfill\includegraphics[width=0.17\textwidth]{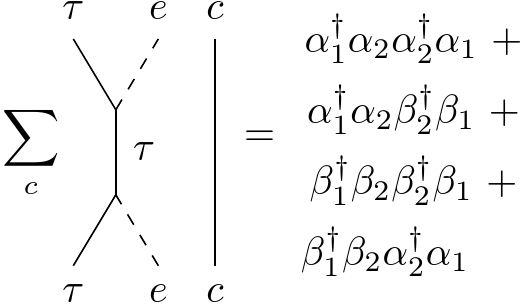}
    \caption{Local observable terms in $1,2$ with global charge $\tau$, up to hermitian conjugation}
    \label{fig:localtau}
\end{figure}

\section{Fibonacci commutation relations}
\label{sec:commrels}

We have defined the Fibonacci creation and annihilation operators in terms of the diagrammatic formalism we have for non-abelian anyons. A future avenue for research is to give a completely algebraic characterisation of Fibonacci anyons. To do so, we need to specify the algebraic relations the Fibonacci creation and annihilation operators follow and which operators one can specify as polynomials of the creation and annihilation operators are observables.

Such a goal is quite ambitious and difficult, being out of the scope of this publication. Nevertheless, we feel that presenting some algebraic relations satisfied by the Fibonacci creation and annihilation operators may be helpful to provide initial insight into such a task and help in becoming familiar with manipulating expressions where the creation and annihilation operators are present.

We can find the following relations for the operators of a single mode: 
\begin{gather}
    \left(\alpha_S\right)^2 =0 \quad \alpha_S \beta_S=\beta_S\alpha_S=0 \quad \alpha_S\alpha_S^\dagger=\beta_S\beta_S^\dagger  \\ \alpha_S\beta_S^\dagger \beta_S= \alpha_S\beta_S^\dagger \alpha_S=\beta_S\alpha_S^\dagger \alpha_S=\beta_S\alpha_S^\dagger \beta_S \\ \alpha_S\alpha_S^\dagger \alpha_S=\alpha_S-\beta_S\alpha_S^\dagger\alpha_S \\ \beta_S\beta_S^\dagger \beta_S=\beta_S-\alpha_S\beta_S^\dagger\beta_S \\ \beta_S^\dagger \beta_S+\alpha_S^\dagger \alpha_S  + \alpha_S\alpha_S^\dagger + \alpha_S\beta_S^\dagger \alpha_S \beta_S^\dagger=\mathbb{I}
\end{gather}

Thanks to the relations above, we can see that any single-mode annihilation and creation operator polynomial reduces to a fourth-degree polynomial at most. The algebraic relations between creation and annihilation operators at different lattice sites are much more difficult to express in simple algebraic equations. It is exciting to see that the annihilation operators do not satisfy equations of the form $\alpha_A \alpha_B = q \alpha_B \alpha_A$ where $q$ is a complex number. In fact, $\alpha_A \alpha_B$ and $\alpha_B \alpha_A$ have disjoint support.

We hope further work is done on studying the algebraic characterization of non-abelian anyons.

\section{Fibonacci Fock states}
\label{sec:fock}

We can now express any state of the canonical basis as a well-ordered sequence of creation operators acting on $\ket{0}$. We present the concrete expressions for three-mode Fibonacci anyons. The expressions for the canonical basis are in Figure \ref{fig:fock}. 

\begin{widetext}
    
\begin{figure}[H]
    \centering
         \includegraphics[width=0.36\textwidth]{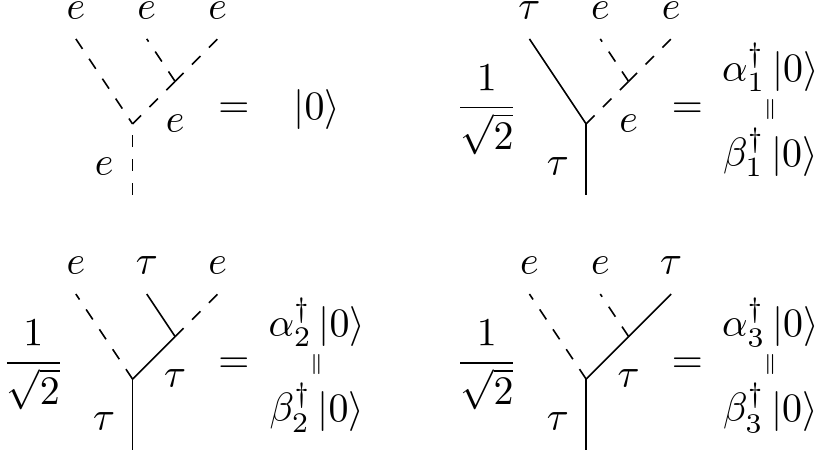}
     \hfill
     \newline
     \newline
     \hfill
         \includegraphics[width=0.5\textwidth]{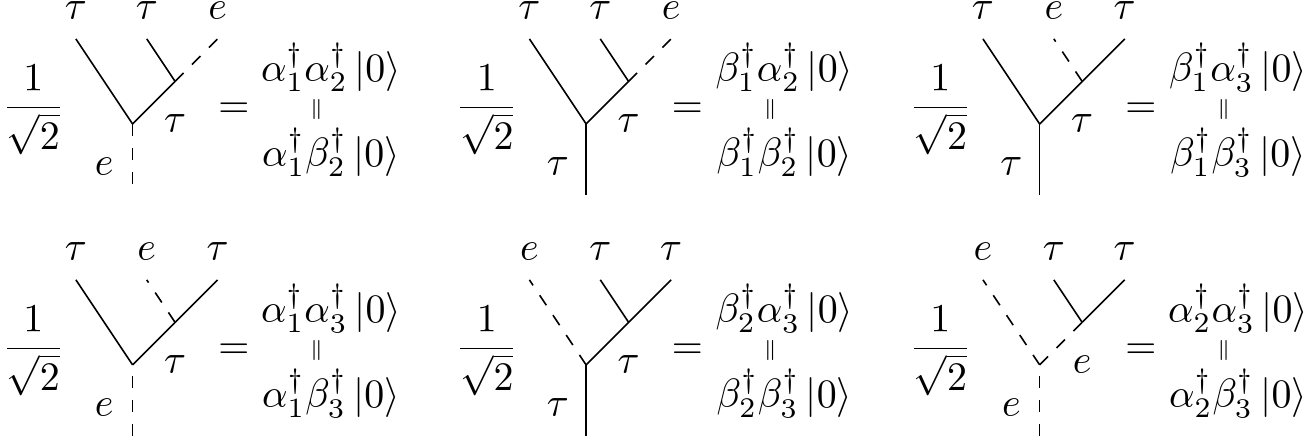}
     \newline
     \newline
     \hfill
         \includegraphics[width=0.33\textwidth]{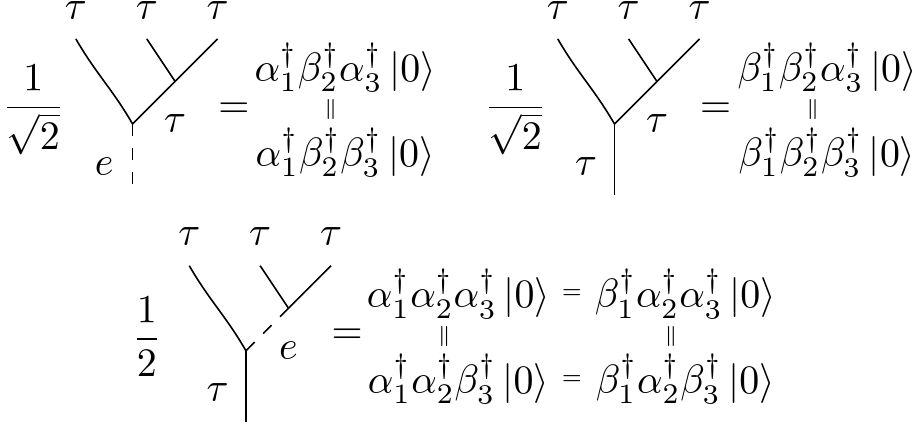}
     \hfill 
    \caption{Canonical basis as a Fock basis, applying the renormalised anyonic creation operators $\alpha,\beta$ to the vacuum.}
    \label{fig:fock}
\end{figure} 

\end{widetext}

\end{document}